%% file: paper.tex
\documentclass[runningheads]{llncs}
\input{preamble}
\usepackage{fullpage}
\title{From LTL to rLTL Monitoring:\\Improved Monitorability through Robust Semantics}

\titlerunning{From LTL to rLTL Monitoring}
 
\author
{
	Corto Mascle\inst{1} \and
	Daniel Neider\inst{2} \and
	Maximilian Schwenger\inst{3} \and\\
	Paulo Tabuada\inst{4} \and
	Alexander Weinert\inst{5} \and
	Martin Zimmermann\inst{6}
}

\authorrunning{C.\ Mascle et al.}

\institute{
ENS Paris-Saclay, Cachan, France\\
\and
Safety and Explainability of Learning Systems Group, Carl von Ossietzky University of Oldenburg, Germany\\
(This work was partly conducted at the Max Planck Institute for Software Systems, Kaiserslautern, Germany) \\
\email{daniel.neider@uol.de}
\and
Reactive Systems Group, Saarland University, Saarbrücken, Germany \\
\and
Department of Electrical and Computer Engineering, UCLA, Los Angeles, USA\\
\and
German Aerospace Center (DLR), Cologne, Germany\\
\and
University of Liverpool, Liverpool, United Kingdom\\
}
\begin{document}

\maketitle

\hyphenation{au-to-ma-ta}

\begin{abstract}
Runtime monitoring is commonly used to detect the violation of desired properties in safety critical cyber-physical systems by observing its executions.
Bauer et al.\ introduced an influential framework for monitoring Linear Temporal Logic~(LTL) properties based on a three-valued semantics for a finite execution: the formula is already satisfied by the given execution, it is already violated, or it is still undetermined, i.e., it can still be satisfied and violated by appropriate extensions of the given execution.
However, a wide range of formulas are not monitorable under this approach, meaning that there are executions for which satisfaction and violation will always remain undetermined no matter how it is extended.
In particular, Bauer et al.\ report that 44\% of the formulas they consider in their experiments fall into this category.

Recently, a robust semantics for LTL was introduced to capture different degrees by which a property can be violated. In this paper we introduce a robust semantics for finite strings and show its potential in monitoring: every formula considered by Bauer et al.\ is monitorable under our approach.
Furthermore, we discuss which properties that come naturally in LTL monitoring --- such as the realizability of all truth values --- can be transferred to the robust setting.
We show that LTL formulas with robust semantics can be monitored by deterministic automata, and provide tight bounds on the size of the constructed automaton. Lastly, we report on a prototype implementation and compare it to the LTL monitor of Bauer et al.\ on a sample of examples.
\end{abstract}

\section{Introduction}
\label{sec:intro}

Runtime monitoring is nowadays routinely used to assess the satisfaction of properties of systems during their execution. To this end, a monitor, a finite-state device that runs in parallel to the system during deployment, evaluates it with respect to a fixed property.
This is especially useful for systems that cannot be verified prior to deployment and, for this reason, can contain hidden bugs.  While it is useful to catch and document these bugs during an execution of a system, we find that the current approach to runtime verification based on Linear Temporal Logic~(\ltl)~\cite{BauerLeuckerSchallhart11} is not sufficiently informative, especially in what regards a system's robustness. Imagine that we are monitoring a property~$\varphi$ and that this property is violated during an execution. In addition to be alerted to the presence of a bug, there are several other questions we would like to have answered such as: Although $\varphi$ was falsified, was there a \emph{weaker} version of $\varphi$ that was still satisfied or did the system fail catastrophically?  Similarly, if we consider a property of the form $\varphi\implies \psi$, where $\varphi$ is an environment assumption and $\psi$ is a system guarantee, and the environment violates $\varphi$ \emph{slightly} along an execution can we still guarantee that $\psi$ is only \emph{slightly} violated? 

Answering these questions requires a logical formalism for specifying properties that provides meaning to terms such as \emph{weaker} and \emph{slightly}. Formalizing these notions within temporal logic, so as to be able to reason about the robustness of a system, was the main impetus behind the definition of \emph{robust Linear-time Temporal Logic} (\rltl)~\cite{DBLP:conf/csl/TabuadaN16}. While reasoning in \ltl yields a binary result, \rltl adopts a five-valued semantics representing different \emph{shades of violation}. Consider, for example, the specification $\Box a \implies \Box b$ requiring that $b$ is always satisfied provided $a$ is always satisfied. In \ltl, if the premise~$a$ is violated in a single position of the trace, then the specification is satisfied vacuously, eliminating all requirements on the system regarding $\Box b$. In this case, \rltl detects a mild violation of the premise and thus allows for a mild violation of the conclusion.

While recent work covers the synthesis~\cite{DBLP:conf/csl/TabuadaN16} and verification problem~\cite{DBLP:conf/hybrid/AnevlavisNPT19,DBLP:conf/cdc/AnevlavisPNT18,DBLP:conf/csl/TabuadaN16} for \rltl, the runtime verification problem is yet to be addressed, except for a preliminary version of the results in this paper presented in the 2020 International Conference on Hybrid Systems: Computation and Control~\cite{DBLP:conf/hybrid/MascleNSTW020}.
Since runtime verification can only rely on finite traces by its nature, interesting theoretical questions open up for \rltl with finite semantics. On the practical side, the very same reasons that make runtime verification for \ltl so useful also motivate the need for developing a finite semantics suitable for \rltl runtime verification. To this end, we tackle the problem of evaluating a property over infinite traces based on a finite prefix similarly to Bauer~\etal~\cite{BauerLeuckerSchallhart11}. If the available information is insufficient to declare a specification violated or satisfied, the monitor reports a $\ukno$. This concept is applied to each degree of violation of the \rltl semantics. Thus, the \rltl monitor's verdict consists of four three-valued bits, as the \rltl semantics is based on four two-valued bits. Each bit represents a degree of violation of the specification in increasing order of severity. 

As an example, consider an autonomous drone that may or may not be in a stable state\footnote{By this we mean, \eg, that the error in tracking a desired trajectory is below a certain threshold.}. The specification requires that it remains stable throughout the entire mission. However, if the take-off is shaky due to bad weather, the drone is unstable for the first couple of minutes. An \ltl monitor thus jumps to the conclusion that the specification is violated whereas an \rltl monitor only reports a \emph{partial} violation. As soon as the drone stabilizes, the \ltl monitor does not indicate any improvement while the \rltl monitor refines its verdict to also report a partial satisfaction.

Some interesting properties that come naturally with \ltl monitoring cannot be seamlessly lifted to \rltl monitoring. 
While it is obvious that all three truth values for finite trace \ltl, \ie, satisfied, violated, and unknown, can be realized for some prefix and formula, the same does not hold for \rltl. 
Intuitively, the second and third bit of the \rltl monitor's four-bit output for the property~$\Boxdot a$ represent whether $a$ eventually holds forever or whether it holds infinitely often, respectively. Based on a prefix, a monitor cannot distinguish between these two shades of violation, rendering some monitor outputs unrealizable.

In addition to that, we investigate how the level of informedness of an \ltl monitor relates to the one of an \rltl monitor. The first observation is that a verdict of an \ltl monitor can be refined at most once, from an unknown to either true or false. With \rltl semantics, however, a monitor can refine its output for a given formula up to four times. 
Secondly, an \ltl monitor can only deliver meaningful verdicts for \emph{monitorable}~\cite{Bauer:2010:CLS:1830408.1830410} properties. Intuitively, a property is monitorable if every prefix can be extended by a finite continuation that gives a definite verdict.
We adapt the definition to robust monitoring and show that neither does \ltl monitorability imply \rltl monitorability, nor vice versa. 

Notwithstanding the above, empirical data suggests that \rltl monitoring indeed provides more information than \ltl monitoring: This paper presents an algorithm synthesizing monitors for \rltl specifications. An implementation thereof allows us to validate the approach by replicating the experiments of Bauer~\etal~\cite{Bauer:2010:CLS:1830408.1830410}. As performance metric, we use \ltl and \rltl monitorability. 
While 44\% of the formulas considered by Bauer~\etal~\cite{Bauer:2010:CLS:1830408.1830410} are not \ltl-monitorable, we show all of them to be \rltl-monitorable.
This indicates that \rltl monitoring is an improvement over \ltl monitoring in terms of monitorability and complements the theoretical results with a practical validation.

This paper is an extended version of the work presented in the 2020 International Conference on Hybrid Systems: Computation and Control~\cite{DBLP:conf/hybrid/MascleNSTW020}. 
The main research contributions are a finite trace semantics for \rltl coupled with an investigation of its properties when compared to \ltl, as well as an algorithm to synthesize monitors for \rltl specifications. Our construction is doubly-exponential in the size of the formula, showing that \rltl monitoring is no more costly than \ltl monitoring.
In addition to the original work~\cite{DBLP:conf/hybrid/MascleNSTW020}, this article features
\begin{enumerate*}[label={(\roman*)}]
	\item a more detailed discussion of the properties of our finite trace semantics for rLTL,
	\item a new running example detailing each step of the monitor construction,
	\item a new example illustrating the nesting of rLTL operators,
	\item refined complexity bounds on our monitor construction, and
	\item all proofs omitted from the conference paper, which provide important additional insight into the problem of monitoring \rltl\ properties.
\end{enumerate*}

\subsection*{Related Work}

In runtime verification~\cite{Drunsinski2000,Havelund2004,Lee99runtimeassurance,MRS17} the specification is often given in \ltl~\cite{Manna:1995:TVR:211468}. 
While properties arguing about the past or current state of a system are always monitorable~\cite{DBLP:conf/tacas/HavelundR02}, \ltl can also express assumptions on the future that cannot be validated using only a finite prefix of a word.
Thus, adaptations of \ltl have been proposed which include different notions of a next step on finite words~\cite{Eisner:2003,Maler:1995:TAA:646702.701843}, lifting \ltl to a three- or four-valued domain~\cite{Bauer:2010:CLS:1830408.1830410,BauerLeuckerSchallhart11}, or applying predictive measures to rule out impossible extensions of words~\cite{Zhang:2012}.

Non-binary monitoring has also been addressed by adding quantitative measures such as counting events~\cite{DBLP:conf/fm/BarringerFHRR12,DBLP:conf/rv/MedhatBFJ16}. 
Most notably, Bartocci~\etal~\cite{DBLP:conf/cav/BartocciBNR18} evaluate the ``likelihood'' that a satisfying or violating continuation will occur.  
To this end, for a given prefix, they count how long a continuation needs to be such that the specification is satisfied/violated; these numbers are then compared against each other.
The resulting verdict is quinary: satisfying/violating, presumably satisfying/violating, or inconclusive.
This approach is similar in nature to our work as it assesses the degree of satisfaction or violation of a given prefix.  
However, the motivation and niche of both approaches differs:  
Bartocci~\etal's approach computes --- intuitively speaking --- the amount of work that is required to satisfy or violate a specification, which allows for estimating the likelihood of satisfaction.
Our approach, however, focuses on measuring the extent to which a specification was satisfied or violated.

Apart from that, monitoring tools collecting statistics~\cite{DBLP:conf/cmsb/AbbasRBSG17,DBLP:conf/esop/AlurFR16,Finkbeiner2005} become increasingly popular:
Snort~\cite{Roesch:1999:SLI:1039834.1039864} is a commercial tool for rule-based network monitoring and computing efficient statistics, Beep Beep 3~\cite{DBLP:conf/rv/Halle16} is a tool based on a query language allowing for powerful aggregation functions and statistical measures. On the downside, these tools impose the overhead of running a heavy-weight application on the monitored system. In contrast, we generate monitor automata out of an \rltl formula. 
Such an automaton can easily and automatically be implemented on almost any system with statically determined memory requirements and negligible performance overhead. 
Similarly, the Copilot~\cite{DBLP:conf/rv/PikeGMN10} framework based on synchronous languages~\cite{DBLP:conf/popl/CaspiPHP87,DBLP:conf/time/DAngeloSSRFSMM05} transforms a specification in a declarative data-flow language into a C implementation of a monitor with constant space and time requirements. 
Lola~\cite{Adolf:2017,DBLP:conf/time/DAngeloSSRFSMM05} allows for more involved computations, also incorporating parametrization~\cite{DBLP:conf/rv/FaymonvilleFST16} and real-time capabilities~\cite{Faymonville19} while retaining constant space and time requirements.

Another approach is to enrich temporal logics with quantitative measures such as taking either the edit distance~\cite{DBLP:journals/fmsd/JaksicBGNN18}, counting the number of possible infinite models for \ltl~\cite{DBLP:conf/atva/FinkbeinerT17,DBLP:journals/acta/TorfahZ18}, incorporating aggregation expressions into metric first-order temporal logic~\cite{DBLP:journals/fmsd/BasinKMZ15}, or using averaging temporal operators that quantify the degree of satisfaction of a signal for a specification by integrating the signal w.r.t.\ a constant reference signal~\cite{DBLP:conf/cav/AkazakiH15}.

Rather than enriching temporal logics with such strong quantitative measures, we consider a robust version of \ltl: \rltl\cite{DBLP:conf/hybrid/AnevlavisNPT19,DBLP:conf/cdc/AnevlavisPNT18,DBLP:journals/tocl/AnevlavisPNT22,DBLP:conf/csl/TabuadaN16}. 
Robust semantics yields information about to which degree a trace violates a property. 
We adapt the semantics to work with finite traces by allowing for intermediate verdicts. 
Here, a certain degree of violation can be classified as ``indefinite'' and refined when more information becomes available to the monitor. 
Similarly, for Signal Temporal Logic~\cite{DBLP:conf/birthday/MalerNP08,DBLP:conf/formats/MalerN04}, Fainekos~\etal~\cite{FAINEKOS20094262} introduced a notion of spacial robustness based on interpreting atomic propositions over the real numbers. The sign of the real number provides information about satisfaction/violation while its absolute value provides information about robustness, i.e., how much can this value be altered without changing satisfaction/violation. This approach is complementary to ours since the notion of robustness in \rltl is related to the temporal evolution of atomic propositions which are interpreted classically, i.e, over the Booleans. Donze~\etal~\cite{DBLP:conf/cav/DonzeFM13} introduced a notion of robustness closer to rLTL in the sense that it measures how long we need to wait for the truth value of a formula to change. 
For this, Cralley~\etal~\cite{tltk} presented a convenient toolbox, achieving high efficiency through parallel evaluation.
While the semantics of \rltl does not allow for quantifying the exact delay needed to change the truth value of a formula, it allows for distinguishing between the influence that different temporal evolutions, e.g., delays, persistence, and recurrence, have on the truth value of an LTL formula. 
Closer to \rltl is the work of Radionova~\etal~\cite{Rodionova:2016:TLF:2883817.2883839} (see also~\cite{DBLP:conf/atva/SilvettiNBB18}) that established an unexpected connection between LTL and filtering through a quantitative semantics based on convolution with a kernel. By using different kernels, one can express weaker or stronger interpretations of the same formula. However, this requires the user to choose multiple kernels and to use multiple semantics to reason about how the degradation of assumptions leads to the degradation of guarantees. In contrast, no such choices are required in \rltl.
Finally, it is worth mentioning that extensions similar to \rltl have been proposed for other temporal logics, such as prompt LTL and linear dynamic logic~\cite{DBLP:journals/corr/abs-1909-08538,NEIDER2021104810}.

Another venue for robust monitoring is machine learning.
Cheng~\cite{robustactivation} presents an algorithm for generating monitors evaluating the distance between the input of a neural net and its training data.
While neural nets are prone to fragility, the monitor is provably robust in the sense that minor input deviations invariably lead to minor changes in the output.
Similarly, Finkbeiner~\etal~\cite{robustrtlola} generate monitors for medical cyber-physical systems controlled by machine learned components.
Due to the complexity of the underlying specification language, they opt for the simpler task of analyzing the robustness of the specification instead.
If the specification is robust, then so will be the generated monitors.

\section{Robust Linear Temporal Logic}
\label{sec:defs}
Throughout this work, we assume basic familiarity with classical \ltl and refer the reader to a textbook for more details on the logic (see, \eg,~\cite{DBLP:books/daglib/0020348}).
Moreover, let us fix some finite set~$P$ of atomic propositions throughout the paper and define~$\Sigma = 2^P$.
We denote the set of finite and infinite words over~$\Sigma$ by~$\Sigma^*$ and~$\Sigma^\omega$, respectively. The empty word is denoted by $\varepsilon$ and  $\sqsubseteq$ and $\sqsubset$ denote the non-strict and the strict prefix relation, respectively.
Moreover, we denote the set of Booleans by~$\bools = \set{0,1}$.

The logics \ltl and \rltl share the same syntax  save for a dot superimposed on temporal operators.
More precisely, the syntax of \rltl is given by the grammar
\[
\varphi \coloneqq
	p  \mid
	\neg \varphi \mid
	\varphi \lor \varphi \mid
	\varphi \land \varphi \mid
	\varphi \Rimplies \varphi 
	\mid \Xdot \varphi \mid
	\varphi \Udot \varphi \mid
	\varphi \Rdot \varphi \mid
	\Diamonddot \varphi \mid
	\Boxdot \varphi,
\]
where $p$ ranges over atomic propositions in $P$ and the temporal operators $\Xdot$, $\Udot$, $\Rdot$, $\Diamonddot$ and $\Boxdot$ correspond to ``next'', ``until'', ``release'', ``eventually'', and ``always'', respectively.%
\footnote
{
	Note that we include the operators $\land$, $\Rimplies$, and $\Rdot$ explicitly in the syntax as they cannot be derived from other operators due to the many-valued nature of \rltl.
	Following the original work on \rltl~\cite{DBLP:conf/csl/TabuadaN16}, we also include the operators $\Diamonddot$ and $\Boxdot$ explicitly (which can be derived from $\Udot$ and $\Rdot$, respectively).
}
The size~$\card{\varphi}$ of a formula~$\varphi$ is the number of its distinct subformulas.
Furthermore, we denote the set of all \ltl and \rltl formulas over $P$ by $ \ltlformulas $ and $ \rltlformulas $, respectively.

The development of  \rltl was motivated by the observation that the difference between ``minor'' and ``major'' violations of a formula cannot be adequately described in a two-valued semantics. If an  \ltl  formula~$\varphi$, for example, demands that the property~$p$ holds at all positions of a word~$\sigma \in \Sigma^\omega$, then~$\sigma$ violates~$\varphi$ even if~$p$ does not hold at only a single position, a very minor violation.
The semantics of \ltl, however, does not differentiate between the~$\sigma$ above and a~$\sigma'$ in which the property~$p$ never holds, a major violation of the property~$\varphi$.

In order to alleviate this shortcoming, Tabuada and Neider introduced Robust Linear-time Temporal Logic (\rltl)~\cite{DBLP:conf/csl/TabuadaN16}, whose semantics allows for distinguishing various ``degrees'' to which a word violates a formula.
More precisely, the semantics of  \rltl are defined over the set $\bools_4 = \{ 0000, 0001, 0011, 0111, 1111 \}$ of five \emph{truth values}, each of which is a monotonically increasing sequence of four bits. We order the truth values in $\bools_4$ by $0000 < 0001 < 0011 < 0111 < 1111$.

Intuitively, this order reflects increasingly desirable outcomes. If the specification is $\Boxdot p$, the least desirable outcome, represented by $0000$, is that $p$ never holds on the entire trace. A slightly more desirable outcome is that $p$ at least holds \emph{sometime} but not infinitely often, which results in the value~$0001$.
An even more desirable outcome would be if $p$ holds infinitely often, while also being violated infinitely often, represented by $0011$. 
Climbing up the ladder of desirable outcomes, the next best one requires $p$ to hold infinitely often while being violated only finitely often, represented by the value~$0111$. 
Lastly, the optimal outcome fully satisfies $\Box p$, so $p$ holds the entire time, represented by $1111$. 
Thus, the first bit states whether $\Box p$ is satisfied, the second one stands for $\Diamond\Box p$, the third one for $\Box\Diamond p$, and the fourth one for $\Diamond p$. If all of them are $0$, $\Box\neg p$ holds. The robust release is defined analogously.

The robust eventually-operator considers future positions in the trace and returns the truth value with the least degree of violation, which is a maximization with respect to the order defined above. This closely resembles the \ltl definition. The robust until is defined analogously.

Based on this, the boolean conjunction and disjunction are defined as $\min$ and $\max$, respectively, w.r.t.\ the order defined above, which generalizes the classical definition thereof. 
For the implication, consider a specification $\Boxdot a \Rimplies \Boxdot g$, where $\Boxdot a$ is an assumption on the environment and $\Boxdot g$ is a system guarantee. 
If the truth value of $\Boxdot g$ is greater or equal to the one of $\Boxdot a$, the implication is fully satisfied.
Thus, the \rltl semantics takes the violation of the assumption into account and lowers the requirements on the guarantees. 
However, if the guarantee exhibits a greater violation than the assumptions, the truth value of the implication is the same as the one of the guarantee.
Lastly, the intuition behind the negation is that every truth value that is not $1111$ constitutes a violation of the specification. Thus, the negation thereof is a full satisfaction ($1111$). The negation of the truth value representing a perfect satisfaction ($1111$) is a full violation ($0000$).

To introduce the semantics, we need some additional notation: For a word $\sigma = \sigma(0) \sigma(1) \sigma(2) \cdots \in \Sigma^\omega$ and a natural number~$n$, define $\suff{\sigma}{n} = \sigma(n) \sigma(n+1) \sigma(n+2) \cdots$, (\ie, as the suffix of $\sigma$ obtained by removing the first $n$ letters of $\sigma$). 
To be able to refer to individual bits of an \rltl truth value $\beta \in \bools_4$, we use $\beta[i]$ with $i \in \{ 1, \ldots, 4 \}$ as to denote the $i$-th bit of $\beta$.

For the sake of a simpler presentation, we denote the semantics of both \ltl and \rltl not in terms of satisfaction relations but by means of \emph{valuation functions}.
For \ltl, the valuation function~$\ltleval \colon \Sigma^\omega \times  \ltlformulas  \to \bools$ assigns to each infinite word $\sigma \in \Sigma^\omega$ and each \ltl formula~$\varphi \in  \ltlformulas $ the value $1$ if $\sigma$ satisfies $\varphi$ and the value $0$ if $\sigma$ does not satisfy $\varphi$, and is defined as usual (see, e.g.,~\cite{DBLP:books/daglib/0020348}).
The semantics of \rltl, on the other hand, is more complex and formalized next by an valuation function~$\rltleval \colon \Sigma^\omega \times  \rltlformulas  \to \bools_4$ mapping an infinite word~$\sigma \in \Sigma^\omega$ and an \rltl formula~$\varphi$ to a truth value in~$ \bools_4$.
\begin{itemize}
\item $\rltlalteval(\sigma,p) =\begin{cases}
	1111 &\text{if $p \in \sigma(0)$,}\\
	0000 &\text{if $p \notin \sigma(0)$,}
\end{cases}$
\item $\rltlalteval(\sigma,\neg \varphi) = \begin{cases} 
 	1111 &\text{if $\rltlalteval(\sigma,\varphi) \neq 1111$,}\\ 
 	0000 &\text{if $\rltlalteval(\sigma,\varphi) = 1111$,} 
 \end{cases}$	

\item
$\rltlalteval(\sigma, \varphi_1 \wedge \varphi_2) = \min\set{\rltlalteval(\sigma, \varphi_1), \rltlalteval(\sigma,\varphi_2) }$, 
\item $\rltlalteval(\sigma, \varphi_1 \vee \varphi_2) = \max\set{\rltlalteval(\sigma, \varphi_1), \rltlalteval(\sigma,\varphi_2) }$, 

\item
$\rltlalteval(\sigma, \varphi_1 \Rimplies \varphi_2) = \begin{cases}
 	1111 &\text{if $\rltlalteval(\sigma,\varphi_1) \le \rltlalteval(\sigma, \varphi_2)$,}\\
 	\rltlalteval(\sigma, \varphi_2) &\text{if $\rltlalteval(\sigma,\varphi_1) > \rltlalteval(\sigma, \varphi_2)$,}
 \end{cases}$
 
 \item $\rltlalteval(\sigma, \Xdot \varphi) = \rltlalteval(\suff{\sigma}{1}, \varphi)$,
 
 \item $\rltlalteval(\sigma, \Diamonddot \varphi ) = \beta$ with $\beta[i] = \max_{n \ge 0} \rltlalteval(\suff{\sigma}{n}, \varphi )[i]$ for $i \in \{ 1, \ldots, 4 \}$,
 \item $\rltlalteval(\sigma, \Boxdot \varphi) = \beta $ with 
 	\begin{align*}
 		\beta[1] & = \min_{n \ge 0} \rltlalteval(\suff{\sigma}{n}, \varphi )[1],\\
 		\beta[2] & = \max_{m \ge 0} \min_{n \ge m}\rltlalteval(\suff{\sigma}{n}, \varphi )[2],\\
 		\beta[3] & = \min_{m \ge 0} \max_{n \ge m}\rltlalteval(\suff{\sigma}{n}, \varphi )[3],\\
 		\beta[4] & = \max_{n \ge 0} \rltlalteval(\suff{\sigma}{n}, \varphi )[4],
 	\end{align*}

 \item $\rltlalteval(\sigma, \varphi_1 \Udot \varphi_2 ) = \beta$ with 
 	\[ \beta[i] = \max_{n \ge 0} \min\{ \rltlalteval(\suff{\sigma}{n}, \varphi_2 )[i], \min_{0\le n' < n} \rltlalteval(\suff{\sigma}{n'}, \varphi_1)[i] \}, \]
 	for $i \in \{ 1, \ldots, 4 \}$,

 \item $\rltlalteval(\sigma, \varphi_1 \Rdot \varphi_2) = \beta $ with 
 	\begin{align*}
		\beta[1] & = \min_{n \ge 0}  \max\set{ \rltlalteval(\suff{\sigma}{n}, \varphi_2 )[1], \max_{0 \le n' < n} \rltlalteval(\suff{\sigma}{n'}, \varphi_1)[1]  },\\
 		\beta[2] & = \max_{m \ge 0} \min_{n \ge m} \max\set{ \rltlalteval(\suff{\sigma}{n}, \varphi_2 )[2], \max_{0 \le n' < n} \rltlalteval(\suff{\sigma}{n'}, \varphi_1)[2]  },\\
 		\beta[3] & = \min_{m \ge 0} \max_{n \ge m} \max\set{ \rltlalteval(\suff{\sigma}{n}, \varphi_2 )[3], \max_{0 \le n' < n} \rltlalteval(\suff{\sigma}{n'}, \varphi_1)[3]  }, \text{ and} \\
 		\beta[4] & = \max_{n \ge 0} \max\set{ \rltlalteval(\suff{\sigma}{n}, \varphi_2 )[4], \max_{0 \le n' < n} \rltlalteval(\suff{\sigma}{n'}, \varphi_1)[4]  }.
 	\end{align*}
\end{itemize}

So as to not clutter this section too much, we refer the reader to the original work by Tabuada and Neider~\cite{DBLP:conf/csl/TabuadaN16} for a thorough introduction and motivation to the preceding semantics.
However, we here want to illustrate the definition above and briefly argue that it indeed captures the intuition described at the beginning of this section.
To this end, we reconsider the formulas $\Boxdot p$, $\Boxdot a \Rimplies \Boxdot g$, $\Boxdot(q \Rimplies \Diamonddot p)$ in Examples~\ref{ex:rltl-semantics-1}, \ref{ex:rltl-semantics-2}, and \ref{ex:rltl-semantics-3} respectively.

\begin{example} \label{ex:rltl-semantics-1}
Consider the formula $\Boxdot p$ and the following five infinite words over the set $P = \{ p \}$ of atomic propositions:
\begin{align*}
	\sigma_1 & = \{ p \}^\omega & \text{(``$p$ holds always'')} \\
	\sigma_2 &= \emptyset \{ p \}^\omega & \text{(``$p$ holds almost always'')} \\
	\sigma_3 & = (\emptyset \{ p \})^\omega & \text{(``$p$ holds infinitely often'')} \\
	\sigma_4 & = \{ p \} \emptyset^\omega & \text{(``$p$ holds finitely often'')} \\
	\sigma_5 & = \emptyset^\omega & \text{(``$p$ holds never'')}
\end{align*}

Let us begin the example with the word $\sigma_1 =  \{ p \}^\omega$.
It is not hard to verify that $\rltlalteval(\sigma_1, \Boxdot p)[1] = 1$ because $p$ always holds in $\sigma_1$, i.e., $\min_{n \ge 0} \rltlalteval(\suff{\sigma}{n}, p)[1] = 1$ for $n \geq 0$.
Using the same argument, we also have $\rltlalteval(\sigma_1, \Boxdot p)[2] = \rltlalteval(\sigma_1, \Boxdot p)[3] =  \rltlalteval(\sigma_1, \Boxdot p)[4] = 1$.
Thus, $\rltlalteval(\sigma_1, \Boxdot p) = 1111$.

As another example, consider the word $\sigma_2 = \emptyset \{ p \}^\omega$.
In this case, we have $\rltlalteval(\sigma_1, \Boxdot p)[1] = 0$ because $\rltlalteval(\suff{\sigma}{0}, p)[1] = 0$ ($p$ does not hold in the first symbol of $\sigma_2$).
However, $\rltlalteval(\sigma_1, \Boxdot p)[2] = 1$ because $p$ holds almost always, i.e., $\max_{m \ge 0} \min_{n \ge m}\rltlalteval(\suff{\sigma}{n}, p)[2] = 1$.
Moreover, $\rltlalteval(\sigma_1, \Boxdot p)[3] =  \rltlalteval(\sigma_1, \Boxdot a)[4] = 1$ and, therefore, $\rltlalteval(\sigma_2, \Boxdot p) = 0111$.
Similarly, we obtain $\rltlalteval(\sigma_3, \Boxdot p) = 0011$, $\rltlalteval(\sigma_4, \Boxdot p) = 0001$, and $\rltlalteval(\sigma_5, \Boxdot p) = 0000$.

In conclusion, this indeed illustrates that the semantics of the robust always is in accordance with the intuition provided at the beginning of this section.
\hfill\exampleend
\end{example}

\begin{example} \label{ex:rltl-semantics-2}
Let us now consider the more complex formula $\Boxdot a \Rimplies \Boxdot g$, where we interpret $a$ to be an assumption on the environment of a cyber-physical system and $g$ one of its guarantees.
Moreover, let $\sigma$ be an infinite word over $P = \{ a, g \}$ such that $\rltlalteval(\sigma, \Boxdot a \Rimplies \Boxdot g) = 1111$.
We now distinguish various cases.

First, let us assume that $\sigma$ is such that $\rltlalteval(\sigma, \Boxdot a) = 1111$, i.e., $a$ always holds.
By definition of the robust implication and since $\rltlalteval(\sigma, \Boxdot a \Rimplies \Boxdot g) = 1111$, this can only be the case if $\rltlalteval(\sigma, \Boxdot g) = 1111$.
Thus, the formula $\Boxdot a \Rimplies \Boxdot g$ ensures that if the environment assumption $a$ always holds, so does the system guarantee $g$. 

Next, assume that $\sigma$ is such that $\rltlalteval(\sigma, \Boxdot a) = 0111$, i.e., $a$ does not always hold but almost always.
By definition of the robust implication and since $\rltlalteval(\sigma, \Boxdot a \Rimplies \Boxdot g) = 1111$, this can only be the case if $\rltlalteval(\sigma, \Boxdot g) \geq 0111$.
In this case, the formula $\Boxdot a \Rimplies \Boxdot g$ ensures that if the environment assumption $a$ holds almost always, then the system guarantee $g$ holds almost always or---even better---always. 

It is not hard to verify that we obtain similar results for the cases $\rltlalteval(\sigma, \Boxdot a) \in \{ 0011, 0001, 0000 \}$.
In other words, the semantics of \rltl ensures that the violation of the system guarantee~$g$ is always proportional to the violation of the environment assumption~$a$ (given that $\rltlalteval(\sigma, \Boxdot a \Rimplies \Boxdot g)$ evaluates to $1111$).
Again, this illustrates that the semantics of the implication is in accordance with the intuition provided at the beginning of this section.
\hfill\exampleend
\end{example}

\begin{example}\label{ex:rltl-semantics-3}
As a last example, let us discuss the nesting of temporal operators.
Consider the formula~$\varphi = \Boxdot( q \Rimplies \Diamonddot p )$ where we interpret $q$ as a request and $p$ as a response. 

We have $\rltleval(\suff{\sigma}{n},\Diamonddot p) = 1111$ if $\suff{\sigma}{n}$ contains a response, otherwise we have $\rltleval(\suff{\sigma}{n},\Diamonddot p) = 0000$.
Similarly, we have $\rltleval(\suff{\sigma}{n},q \Rimplies\Diamonddot p) = 1111$ if $q \in \sigma(n)$ implies that $\suff{\sigma}{n}$ contains a response.
On the other hand, if $q \in \sigma(n)$ and $\suff{\sigma}{n}$ does not contain a response then we have $\rltleval(\suff{\sigma}{n},q \Rimplies\Diamonddot p) = 0000$.

From these observations, we can deduce $\rltleval(\sigma,\varphi) = 1111$ if every request in $\sigma$ is followed by a response, which is equivalent to the \ltl formula~$\varphi_1 = \Box(q \implies \Diamond p)$ that expresses a request-response property.
Further, we have $\rltleval(\sigma,\varphi) = 0111$ if and only if~$\sigma$ violates~$\varphi_1$ and if from some point onwards, every request in $\sigma$ is followed by a response.
This is equivalent to the \ltl formula~$\neg\varphi_1 \land \varphi_2$ with $\varphi_2 = (\Box\Diamond q) \implies (\Box\Diamond p)$, which expresses strong fairness.
Similarly, we have $\rltleval(\sigma,\varphi) = 0011$ if and only if~$\sigma$ violates~$\varphi_2$ and if for infinitely many positions, if there is a request in~$\sigma$ at that position, then it is followed by a response.
This is equivalent to the \ltl formula~$\neg\varphi_2 \land \varphi_3$ with $\varphi_3 = (\Diamond\Box q) \implies (\Box\Diamond p)$, which expresses weak fairness.
Moreover, we have $\rltleval(\sigma,\varphi) = 0001$ if and only if~$\sigma$ violates~$\varphi_3$ and if there is some position such that if there is a request in~$\sigma$ at that position, then it is followed by a response.
This is equivalent to the \ltl formula~$\neg\varphi_3 \land \varphi_4 $ with $\varphi_4= (\Box q) \implies (\Diamond p)$, which expresses a very weak notion of fairness.
Finally, we have $\rltleval(\sigma,\varphi) = 0000$ if and only if~$\sigma$ violates~$\varphi_4$.

For~$i \in \set{1,2,3}$, the \ltl formula $\varphi_i$ implies $\varphi_{i+1}$.
Thus, if a trace~$\sigma$ violates~$\varphi_{i+1}$, it also violates~$\varphi_i$.
This further illustrates the monotonicity of \rltl.
This monotonicity also allows us to only require that~$\varphi_{i+1}$ violates~$\varphi_i$ in the intuitive explanations above, instead of having to require violations of all~$\varphi_{i'}$ with~$i' \leq i$.
\hfill\exampleend	
\end{example}

It is important to note that \rltl is an extension of \ltl.
In fact, the \ltl semantics can be recovered from the first bit of the \rltl semantics (after every implication $\varphi \Rimplies \psi$ has been replaced with $\lnot \varphi \lor \psi$).%
\footnote{
	It turns out that Tabuada and Neider's original proof~\cite[Proposition~5]{DBLP:conf/csl/TabuadaN16} has a minor mistake.
	Although the first bit of the \rltl semantics coincides with the original \ltl semantics for all formulas that do not contain implications, the formula~$\Boxdot \lnot a \implies \Boxdot a$ is an example witnessing this claim is no longer correct in the presence of implications, e.g., for $\set{a}\emptyset^\omega$.
	However, this issue can be fixed by replacing every implication $\varphi \Rimplies \psi$ with $\lnot \varphi \lor \psi$.
	This substitution results in an equivalent \ltl formula for which the first bit of the \rltl semantics indeed coincides with the \ltl semantics.
}

\begin{lemma}[\cite{DBLP:conf/csl/TabuadaN16}, Proposition~5]
\label{lemma:rltlextendsltl}
Let $\varphi$ be an \ltl formula without implications, and let $\varphi'$ be the corresponding \rltl formula (obtained by dotting all temporal operators). Then, we have $\rltleval(\sigma, \varphi')[1] = \ltleval(\sigma, \varphi)$ for every trace~$\sigma$.
\end{lemma}

To reduce the number of cases we have to consider in our inductive proofs (for instance the one for Lemma~\ref{lem:tech}), we note that the robust eventually and the robust always operator are syntactic sugar.  Formally, we say that two \rltl formulas~$\varphi_1, \varphi_2$ are equivalent if $\rltleval(\sigma, \varphi_1) = \rltleval(\sigma, \varphi_2)$ for every $\sigma \in \Sigma^\omega$. Now, let $ \true =  p \lor \neg p$ and $\false = p \land \neg p$ for some atomic proposition~$p$. Then, the robust eventually and the robust always are, as usual, expressible in terms of the robust until and the robust release, respectively.

\begin{remark}
\label{rem:sugar}
\hfill
\begin{enumerate}
	\item $\Diamonddot \varphi$ and $\true \Udot \varphi$ are equivalent.
	\item $\Boxdot \varphi$ and $\false \Rdot \varphi$ are equivalent.
\end{enumerate}
\end{remark}

\subsection{An Alternative Definition of Robust Semantics for LTL}
\label{subsec:altsem}
Before we introduce \rltl monitoring, we need to introduce an alternative definition of the semantics of \rltl, which is more convenient to prove some of the results from Section~\ref{sec:problem}. This alternative definition has been introduced in later works on \rltl~\cite{DBLP:conf/cdc/AnevlavisPNT18,DBLP:conf/hybrid/AnevlavisNPT19}.

\begin{definition} \label{def:rltl-semantics-alt}
Let the function~$\rltlsem \colon \{ 1, \ldots, 4 \} \times  \rltlformulas  \to  \ltlformulas$
be inductively defined as in Table~\ref{tag:rltl-2-ltl}.
The \rltl semantics is then given as the valuation function
$ \rltleval \colon \Sigma^\omega \times  \rltlformulas  \to \bools_4$,
where for every $\sigma \in \Sigma^\omega$, every \rltl formula~$\varphi$, and every $i \in \{ 1, \ldots, 4 \}$, the $i$-th bit of $\rltleval(\sigma, \varphi)$ is defined as
$\rltleval(\sigma, \varphi)[i] = \ltleval \bigl( \sigma, \rltlsem(i, \varphi) \bigr) $
(\ie, via the semantics of the \ltl formulas~$\rltlsem(i, \varphi)$).
\end{definition}

\begin{table*}[t]
	\centering
	\caption{The function $\rltlsem \colon \{ 1, \ldots, 4 \} \times  \rltlformulas  \to  \ltlformulas $.} \label{tag:rltl-2-ltl}

	\begin{tabular}{lcl}
		\toprule
		\emph{Operator} & \emph{Symbol} & \emph{Semantics ($\varphi, \psi \in  \rltlformulas $)} \\
		\toprule
		Atomic & \multirow{2}{*}{$p \in P$} & \multirow{2}{*}{$1 \leq i \leq 4$: $\rltlsem(i, p) = p$} \\
		proposition \\ \midrule 
		Negation & $\lnot$ & $1 \leq i \leq 4$: $\rltlsem(i, \lnot \varphi) \coloneqq \lnot \rltlsem(1, \varphi)$ \\ \midrule
		Disjunction & $\lor$ & $1 \leq i \leq 4$: $\rltlsem(i, \varphi \lor \psi) \coloneqq \rltlsem(i, \varphi) \lor \rltlsem(i, \psi)$ \\ \midrule
		Conjunction & $\land$ & $1 \leq i \leq 4$: $\rltlsem(i, \varphi \land \psi) \coloneqq \rltlsem(i, \varphi) \land \rltlsem(i, \psi)$ \\ \midrule
		\multirow{2}{*}{Implication} & \multirow{2}{*}{$\Rimplies$} & $1 \leq i \leq 3$: $\rltlsem(i, \varphi \Rimplies \psi) \coloneqq (\rltlsem(i, \varphi) \Rimplies \rltlsem(i, \psi)) \land \rltlsem(i+1, \varphi \Rimplies \psi)$; \\
		& & \phantom{$1 \leq i \leq 3$:} $\rltlsem(4, \varphi \Rimplies \psi) \coloneqq \rltlsem(4, \varphi) \Rimplies \rltlsem(4, \psi)$ \\ \midrule
		Robust next & $\Xdot$ & $1 \leq i \leq 4$: $\rltlsem(i, \Xdot \varphi) \coloneqq \X \rltlsem(i, \varphi)$ \\ \midrule
		Robust & \multirow{2}{*}{$\Diamonddot$} & \multirow{2}{*}{$1 \leq i \leq 4$: $\rltlsem(i, \Diamonddot \varphi) \coloneqq \Diamond \rltlsem(i, \varphi)$} \\
		eventually \\ \midrule
		\multirow{2}{*}{Robust always} & \multirow{2}{*}{$\Boxdot$} & $\rltlsem(1, \Boxdot \varphi) \coloneqq \Box \rltlsem(1, \varphi)$; $\rltlsem(2, \Boxdot \varphi) \coloneqq \Diamond\Box \rltlsem(2, \varphi)$; \\
		& & $\rltlsem(3, \Boxdot \varphi) \coloneqq \Box\Diamond \rltlsem(3, \varphi)$; $\rltlsem(4, \Boxdot \varphi) \coloneqq \Diamond \rltlsem(4, \varphi)$ \\ \midrule
		Robust until & $\Udot$ & $1 \leq i \leq 4$: $\rltlsem(i, \varphi \Udot \psi) \coloneqq \rltlsem(i, \varphi) \U \rltlsem(i, \psi)$ \\ \midrule
		\multirow{4}{*}{Robust release} & \multirow{4}{*}{$\Rdot$} & $\rltlsem(1, \varphi \Rdot \psi) \coloneqq \rltlsem(1, \varphi) \R \rltlsem(1, \psi)$; \\
		& & $\rltlsem(2, \varphi \Rdot \psi) \coloneqq \Diamond\Box \rltlsem(2, \psi) \lor \Diamond \rltlsem(2, \varphi)$; \\
		& & $\rltlsem(3, \varphi \Rdot \psi) \coloneqq \Box\Diamond \rltlsem(3, \psi) \lor \Diamond \rltlsem(3, \varphi)$; \\
		& & $\rltlsem(4, \varphi \Rdot \psi) \coloneqq \Diamond \rltlsem(4, \psi) \lor \Diamond \rltlsem(4, \varphi)$ \\
		\bottomrule
	\end{tabular}
\end{table*}

As a consequence of Lemma~\ref{lemma:rltlextendsltl} (cf.~\cite{DBLP:conf/csl/TabuadaN16}, Proposition~5), we know that \rltl is at least as expressive as \ltl. The latter definition of the semantics of \rltl shows that it is not more expressive than \ltl, in the sense that for all \rltl formulas there exist \ltl formulas giving the truth values of each of the four bits.
However, it is more convenient to work with one formula of \rltl than to work with the four \ltl formulas capturing it.

A useful feature of the alternative semantics is the following property: To determine the truth value of an \rltl formula~$\varphi$ on $\sigma$, it suffices to determine the truth values of the \ltl formulas~$\rltlsem(i, \varphi)$ on $\sigma$. For certain formulas, $\rltlsem(i, \varphi)$ is obtained from $\varphi$ by a very simple rewriting, as shown below.

\begin{remark}
\label{rem:intuitivesemantics}
Let $\varphi$ be an \rltl formula that has no always in the scope of a negation and only uses negation, conjunction, disjunction, next, eventually, and always.
Then,
\begin{itemize}
	\item $\rltlsem(1, \varphi)$ is equivalent to the formula obtained from $\varphi$ by replacing every $\Xdot$ by $\X$, every $\Diamonddot$ by $\Diamond$, and every $\Boxdot$ by $\Box$, 

	\item $\rltlsem(2, \varphi)$ is equivalent to the formula obtained from $\varphi$ by replacing every $\Xdot$ by $\X$, every $\Diamonddot$ by $\Diamond$, and every $\Boxdot$ by $\Diamond\Box$, 

	\item $\rltlsem(3, \varphi)$ is equivalent to the formula obtained from $\varphi$ by replacing every $\Xdot$ by $\X$, every $\Diamonddot$ by $\Diamond$, and every $\Boxdot$ by $\Box\Diamond$, and

	\item $\rltlsem(4, \varphi)$ is equivalent to the formula obtained from $\varphi$ by replacing every $\Xdot$ by $\X$, every $\Diamonddot$ by $\Diamond$, and every $\Boxdot$ by $\Diamond$.

\end{itemize} 	
\end{remark}

\section{Monitoring Robust LTL}
\label{sec:problem}
In their work on \ltl monitoring, Bauer et al.~\cite{BauerLeuckerSchallhart11} define the problem of runtime monitoring as \textit{``check[ing] \ltl properties given finite prefixes of infinite [words]''}.
More formally, given some prefix~$u \in \Sigma^*$ and some \ltl formula~$\varphi$, it asks whether all, some, or no infinite extension~$u\sigma \in \Sigma^\omega$ of $u$ by some $\sigma \in \Sigma^\omega$ satisfies~$\varphi$.
To reflect these three possible results, the authors use the set~$\ternaries = \set{0, \ukno, 1}$ to define a three-valued logic that is syntactically identical to \ltl, but equipped with a semantics in form of an evaluation function~$\ltlmoneval  \colon \Sigma^* \times  \ltlformulas  \rightarrow \ternaries$ over finite prefixes.
This semantics is defined such that~$\ltlmoneval(u, \varphi)$ is equal to $0$ (is equal to $1$) if no (if every) extension~$u \sigma$ of $u$ satisfies $\varphi$. 
If neither is the case, \ie, if there is an extension of $u$ that satisfies $\varphi$ and there is an extension of $u$ that does not satisfy $\varphi$, then $\ltlmoneval (u, \varphi)$ is equal to $\ukno$. 

We aim to extend the approach of Bauer et al.\ to \rltl, whose semantics is based on truth values from the set~$\bools_4$ (containing the sequences of length four in $0^*1^*$).
As a motivating example, let us consider the formula~$\varphi = \Boxdot s$\label{runningexample} for some atomic proposition~$s$ and study which situations can arise when monitoring this formula.
Note that the truth value of $\varphi$ can be obtained by concatenating the truth values of the \ltl formulas $\varphi_1 = \Box s$, $\varphi_2 =  \Diamond\Box s$, $\varphi_3 =  \Box\Diamond s$, and $\varphi_4 =  \Diamond s$. 
 
First, consider the empty prefix and its two extensions~$\emptyset^\omega$ and $\set{s}^\omega$. We have $\rltleval(\emptyset^\omega, \varphi) = 0000$ and $\rltleval(\set{s}^\omega, \varphi) = 1111$. 
Thus, all four bits can both be equal to $0$ and $1$. 
This situation is captured by the sequence~$\ukno \ukno \ukno \ukno $ which signifies that for every position~$i$ and every bit~$b \in \bools$, there exists an extension of $\varepsilon$ that has bit~$b$ in the $i$-th position of the truth value with respect to $\varphi$. 
 
Now, consider the prefix~$\set{s}$ for which we have $\rltleval(\set{s}\sigma, \varphi)[4] = 1$ for every $\sigma \in \Sigma^\omega$ as $\varphi_4 = \Diamond s$ is satisfied on each extension of $\set{s}$ ($s$ has already occurred). 
On the other hand, $\rltleval(\set{s}\emptyset^\omega, \varphi) = 0001$ and $\rltleval(\set{s}\set{s}^\omega, \varphi) = 1111$, \ie, the first three bits can both be $0$ and $1$ by picking an appropriate extension. 
Hence, the situation is captured by the sequence~$\ukno \ukno \ukno 1$, signifying that the last bit is determined by the prefix, but the first three are not.
Using dual arguments, the sequence~$0\ukno \ukno \ukno$ is used for the prefix~$\emptyset{}$, signifying that the  first bit is determined by the prefix as every extension violates $\varphi_1 = \Box s$.
However, the last three bits are not yet determined by the prefix, hence the trailing $\ukno$'s.

Finally, consider the prefix~$\set{s}\emptyset$.
Using the same arguments as for the previous two prefixes, we obtain $\rltleval(\set{s}\emptyset\sigma, \varphi)[1] = 0$ and $\rltleval(\set{s}\emptyset\sigma, \varphi)[4] = 1$ for every $\sigma \in \Sigma^\omega$.
Also, as before, we have $\rltleval(\set{s}\emptyset\emptyset^\omega, \varphi) = 0 0 0 1$ and $\rltleval(\set{s}\emptyset\set{s}^\omega, \varphi) = 0 1 1 1$.
Hence, here we obtain the sequence~$0 \ukno \ukno 1$ signifying that the first and last bit  are determined by the prefix, but the middle two are not.

In general, we use truth values of the form~$0^*?^*1^*$, which follows from the fact that the truth values of  \rltl are in $0^*1^*$. Hence, let $\ternaries_4$ denote the set of sequences of length four in $0^*?^*1^*$.
Based on $\ternaries_4$, we now formally define the \rltl monitoring semantics as a bitwise generalization of the  \ltl definition.

\begin{definition} 
\label{def:monitor:output}
The semantics of the robust monitor $\rltlmoneval  \colon \Sigma^* \times  \rltlformulas  \to \bools_4^{\ukno}$ is defined as $	\rltlmoneval (u, \varphi) = \beta$ with 
\[
  \beta[i] = \begin{cases} 0 & \text{if $\rltleval(u\sigma, \varphi)[i] = 0$ for all $\sigma \in \Sigma^\omega$;} \\ 1 & \text{if $\rltleval(u\sigma, \varphi)[i] = 1$ for all $\sigma \in \Sigma^\omega$; and} \\ \ukno & \text{otherwise,} \end{cases}
\]
for every $i \in \{ 1, \ldots, 4 \}$, every \rltl formula~$\varphi$, and every $u \in \Sigma^*$.
\end{definition}

First, let us remark that our notion of \rltl monitoring indeed refines the notion of \ltl monitoring, which follows immediately from Lemma~\ref{lemma:rltlextendsltl}.

\begin{remark}
\label{remark:rltlmonitoringextendsltlmonitoring}
Let $\varphi$ be an \ltl formula without implications, and let $\varphi'$ be the corresponding \rltl formula (obtained by dotting all temporal operators). Then, we have $\rltlmoneval(u, \varphi')[1] = \ltlmoneval(u, \varphi)$ for every $u \in \Sigma^*$.
	
\end{remark}

Using \rltl monitoring semantics, we are able to infer information about the infinite run of a system after having read only a finite prefix thereof.
In fact, this robust semantics provides far more information about the degree of violation of the specification than classical \ltl monitoring as each bit of the monitoring output represents a degree of violation of the specification:
a $\ukno$ turning into a $0$ or $1$ indicates a deterioration or improvement in the system's state, respectively.
Consider, for instance, an autonomous drone with specification $\varphi = \Boxdot s$ where $s$ denotes a state of stable flight (recall the motivating example on Page~\pageref{runningexample}).
Initially, the monitor would output~$\ukno\ukno\ukno\ukno$ due to a lack of information.
If taking off under windy conditions, the state~$s$ is not reached initially, hence the monitor issues a warning by producing $\rltlmoneval (\emptyset^n,\varphi) = 0\ukno\ukno\ukno$ for every $n > 0$.
Thus, the safety condition is violated temporarily, but not irrecoverably.
Hence, mitigation measures can be initiated.
Upon success, the monitoring output turns into $\rltlmoneval (\emptyset^n \set{s}, \varphi) = 0\ukno\ukno1$ for every $n > 0$, signaling that flight was stable for some time.

Before we continue, let us first state that the new semantics is well-defined, \ie, that the sequence~$\beta[1]\beta[2]\beta[3]\beta[4]$ in Definition~\ref{def:monitor:output} is indeed in $\ternaries_4$.

\begin{lemma}
\label{lem:monitorwelldefined}
$\rltlmoneval (u, \varphi) \in \ternaries_4$ for every \rltl formula~$\varphi $ and every $u \in \Sigma^*$.
\end{lemma}

\begin{proof}
Let $\rltlmoneval(u, \varphi)[i]=0$ and $j<i$. By definition of $\rltlmoneval$, we have $\rltleval(u\sigma, \varphi)[i] = 0$ for every $\sigma \in \Sigma^\omega$. Hence, due to the monotonicity of the truth values from $\bools_4$ used to define $\rltleval$, we obtain $\rltleval(u\sigma, \varphi)[j] = 0$ for every such $\sigma$. Hence, $\rltlmoneval(u, \varphi)[j]=0$.  

A dual argument shows that $\rltlmoneval(u, \varphi)[i]=1$ and $j>i$ implies $\rltlmoneval(u, \varphi)[j]=1$. Combining both properties yields $\rltlmoneval(u, \varphi) \in 0^*\ukno^*1^*$, \ie, $\rltlmoneval(u, \varphi) \in \ternaries_4$.
\qed\end{proof}

After having shown that every possible output of $\rltlmoneval $ is in $\ternaries_4$, the next obvious question is whether $\rltlmoneval $ is surjective, \ie, whether every truth value $\beta \in \ternaries_4$ is \emph{realized} by some prefix~$u \in \Sigma^*$ and some \rltl formula~$\varphi $ in the sense that $\rltlmoneval (u, \varphi) = \beta$. Recall the motivating example above: The formula~$\Boxdot s$ realizes at least the following four truth values: $\ukno \ukno \ukno \ukno $ (on $\varepsilon$), $\ukno \ukno \ukno 1$ (on $\set{s}$), $0 \ukno \ukno \ukno $ (on $\emptyset$), and $0 \ukno \ukno 1$ (on $\set{s}\emptyset$). It is not hard to convince oneself that these are all truth values realized by $\Boxdot s$ as they represent the following four types of prefixes that can be distinguished: the prefix is empty (truth value~$\ukno \ukno \ukno \ukno $), the prefix is in $\set{s}^+$ (truth value~$\ukno \ukno \ukno 1$), the prefix is in $\emptyset^+$ (truth value~$0 \ukno \ukno \ukno $), or the prefix contains both an $\set{s}$ and an $\emptyset$ (truth value~$0 \ukno \ukno 1$). 

For most other truth values, it is straightforward to come up with \rltl formulas and prefixes that realize them.
See Table~\ref{tag:truthvalues} for an overview and recall Remark~\ref{rem:intuitivesemantics}, which is applicable to all these formulas.

\begin{table*}[t]
	\centering
	\caption{Realizable truth values. For every truth value~$\beta$, the next two columns show prefixes~$u$ and formulas~$\varphi$ such that $\rltlmoneval(u,\varphi) = \beta$, or that $\beta$ is unrealizable.} 
	\label{tag:truthvalues}
		\renewcommand{\arraystretch}{1.2}
		\setlength{\tabcolsep}{5pt}
	\begin{tabular}{lll@{\hskip 2cm}llll}

		\toprule
		\emph{Value}&\emph{Prefix}&\emph{Formula}&\hspace{.75cm}&\emph{Value}&\emph{Prefix}&\emph{Formula}\\
		\toprule
$0000$ &  $ \varepsilon $ &  $ a \wedge \neg a $   &&$0\ukno 11$ & $\emptyset\set{a} $  &   $\Boxdot a \vee \Boxdot \neg a $\\
$000\ukno$ &  $  \varepsilon $ &  $\Diamonddot\Boxdot a \wedge \Diamonddot\neg \Diamonddot a $ && $0111$ & $\emptyset\set{a}$ & $a\Rdot a$  \\
$0001$ & \multicolumn{2}{l}{unrealizable}   &&$\ukno \ukno \ukno \ukno$ & $\varepsilon $  &  $\Boxdot a $   \\ 
$00\ukno \ukno$  & $\varepsilon $  &  $ \Boxdot a \wedge \Boxdot \neg a$    &&$\ukno \ukno \ukno 1$ & $ \set{a}$ &   $ \Boxdot a$   \\
$00\ukno 1$ & $ \emptyset\set{a}$  &  $\Boxdot a \wedge \Boxdot \neg a $  &&$\ukno \ukno 11$ & $\varepsilon $  &   $\Boxdot a \vee \Diamonddot\neg \Diamonddot a $  \\
$0011$ & \multicolumn{2}{l}{unrealizable}  &&$\ukno 111$ & $ \varepsilon$  &  $ \Boxdot a \vee \neg \Diamonddot\neg \Diamonddot \neg a$  \\
$0\ukno \ukno \ukno$ & $ \emptyset$  &  $\Boxdot a $   && $1111$ &  $\varepsilon $ & $ a \vee \neg a $ \\
$0\ukno \ukno 1$ & $\emptyset\set{a} $  &   $\Boxdot a $ &&& \\
\bottomrule
\end{tabular}  
\end{table*}
 
For others, such as $0011$, it is much harder.
Intuitively, to realize $0011$, one needs to find an \rltl formula~$\varphi$ and a prefix~$u \in \Sigma^*$ such that
 the formula obtained by replacing all $\Boxdot$ in $\varphi$ by $\Diamond\Box$ is not satisfied by any extension of $u$, but
	the formula obtained by replacing all $\Boxdot$ in $\varphi$ by $\Box\Diamond$ is satisfied by every extension of $u$.\footnote{Note that this intuition breaks down in the presence of implications and negation, due to their non-standard definitions.}
	Thus, intuitively, the prefix has to differentiate between a property holding almost always and holding infinitely often.  
It turns out that no such $u$ and $\varphi$ exist.
A similar argument is true for $0001$, leading to the following theorem.

\begin{theorem} 
\label{thm:realizabletruthvalues}
All truth values except for $0011$ and $0001$ are realizable. 
\end{theorem}

The unrealizability results for the truth values $0011$ and $0001$ are based on the following technical lemma (the reader might want to skip the proof for now and consult it at a later time). %

\begin{lemma}
\label{lem:tech}
Let $\varphi$ be an $\rltl$ formula.
Then, the following holds:
\begin{enumerate}

	\item\label{lem:tech:2eq3} $\rltleval(u\emptyset^\omega , \varphi)[2] = \rltleval(u\emptyset^\omega , \varphi)[3]$ for all $u \in \Sigma^*$.

	\item\label{lem:tech:3eq4} $\rltleval(u^\omega , \varphi)[3] = \rltleval(u^\omega , \varphi)[4]$ for all non-empty $u \in \Sigma^*$.

	\item\label{lem:tech:1eq2noreleaseonly} If $\varphi$ does not contain the release operator, then $\rltleval(u^\omega , \varphi)[1] = \rltleval(u^\omega , \varphi)[2]$ for all non-empty $u \in \Sigma^*$.

\end{enumerate}
\end{lemma}

\begin{proof} The proofs of all three items proceed by induction over the construction of $\varphi$. The induction start and the induction steps for Boolean connectives can be abstracted into the following closure property, which follows easily from the original  definition of $\rltleval$ in Section~\ref{sec:defs}:

\begin{quote}
Let $T \subseteq \bools_4$ contain $0000$ and $1111$. If $\rltleval(\sigma, \varphi_1)$ and $\rltleval(\sigma, \varphi_2) $ are in  $T$, then so are $\rltleval(\sigma, p)$ for atomic propositions~$p$, $\rltleval(\sigma, \neg \varphi_1)$, $\rltleval(\sigma,  \varphi_1 \land \varphi_2)$, $\rltleval(\sigma, \varphi_1 \lor \varphi_2)$, and $\rltleval(\sigma, \varphi_1 \Rimplies \varphi_2)$.
\end{quote}

\paragraph{Claim~\ref{lem:tech:2eq3})} The induction start and the induction step for the Boolean operators follow from the closure property, where we pick $T$ to be the set of truth values from $\bools_4$ whose second and third bit coincide. Furthermore, due to Remark~\ref{rem:sugar}, we only have to consider the inductive steps for the next, until, and release operator. All three cases rely on the following simple fact: A suffix~$\suff{u\emptyset^\omega}{n}$ for some $n$ is again of the form~$u'\emptyset^\omega$, \ie, the induction hypothesis is applicable to suffixes. Also, if $n \ge \size{u}$, then $\suff{u\emptyset^\omega}{n} = \emptyset^\omega$. In particular, $u\emptyset^\omega$ has only finitely many distinct suffixes. 

So, first consider a formula of the form~$\varphi = \Xdot \varphi_1$. Then, we have, for an arbitrary~$u \in \Sigma^*$,
\begin{align*}
\rltleval(u\emptyset^\omega, \varphi)[2] = &\rltleval(\suff{u\emptyset^\omega}{1}, \varphi_1)[2]\\
=&\rltleval(\suff{u\emptyset^\omega}{1}, \varphi_1)[3]
=\rltleval(u\emptyset^\omega, \varphi)[3],
\end{align*}
where the second equality is due to the induction hypothesis being applied to the suffix~$\suff{u\emptyset^\omega}{1}$.

Next, consider a formula of the form~$\varphi = \varphi_1 \Udot \varphi_2$. Then, we have, for an arbitrary~$u \in \Sigma^*$,
\begin{align*}
&\rltleval(u\emptyset^\omega, \varphi)[2]\\
=&\max_{n \ge 0} \min\set{\rltleval(\suff{u\emptyset^\omega}{n}, \varphi_2 )[2], \min_{0\le n' < n} \rltleval(\suff{u\emptyset^\omega}{n'}, \varphi_1)[2] }\\
=&\max_{n \ge 0} \min\set{\rltleval(\suff{u\emptyset^\omega}{n}, \varphi_2 )[3], \min_{0\le n' < n} \rltleval(\suff{u\emptyset^\omega}{n'}, \varphi_1)[3] }\\
=&\rltleval(u\emptyset^\omega, \varphi)[3], 
\end{align*}
where the second equality follows from an application of the induction hypothesis to the suffixes~$\suff{u\emptyset^\omega}{n}$ and $\suff{u\emptyset^\omega}{n'}$.

It remains to  consider a formula of the form~$\varphi = \varphi_1 \Rdot \varphi_2$. Then, we have, for an arbitrary~$u \in \Sigma^*$, that $\rltleval(u\emptyset^\omega, \varphi)[2]$ is by definition equal to 
\begin{align*}
&\max_{m \ge 0} \min_{n \ge m} \max\{ \rltleval(\suff{u\emptyset^\omega}{n}, \varphi_2 )[2], \max_{0 \le n' < n} \rltleval(\suff{u\emptyset^\omega}{n'}, \varphi_1)[2]  \}\\
=&\max_{m \ge 0} \min_{n \ge m} \max\set{\rltleval(\suff{u\emptyset^\omega}{n}, \varphi_2 )[3],   \max_{0 \le n' < n} \rltleval(\suff{u\emptyset^\omega}{n'}, \varphi_1)[3] }\\
=&\max_{m \ge \size{u}} \min_{n \ge m} \max\set{  \rltleval(\suff{u\emptyset^\omega}{n}, \varphi_2 )[3],  \max_{0 \le n' < n} \rltleval(\suff{u\emptyset^\omega}{n'}, \varphi_1)[3] }\\
=&\max_{m \ge \size{u}} \min_{n \ge m} \max\set{   \rltleval(\emptyset^\omega, \varphi_2 )[3],  \max_{0 \le n' \le \size{u}} \rltleval(\suff{u\emptyset^\omega}{n'}, \varphi_1)[3] }\\
=&\multicolumn{3}{l}{$\displaystyle\max\set{ \rltleval(\emptyset^\omega, \varphi_2 )[3], \max_{0 \le n' \le \size{u}} \rltleval(\suff{u\emptyset^\omega}{n'}, \varphi_1)[3]}, $}
\end{align*}
The first equality follows from twice applying the induction hypothesis. For the second one, observe that 
\[
\min_{n \ge m} \max\set{ \rltleval(\suff{u\emptyset^\omega}{n}, \varphi_2 )[3], \max_{0 \le n' < n} \rltleval(\suff{u\emptyset^\omega}{n'}, \varphi_1)[3] }\]
 is increasing in $m$. For the third one, note that for all $n\geq \size{u}$, $\suff{u\emptyset^\omega}{n} = \emptyset^\omega$, which means that we have eliminated every occurrence of $m$ and $n$. This explains the last equality. 
Similarly, $\rltleval(u\emptyset^\omega, \varphi)[3]$ is by definition equal to 
\begin{align*}
&\min_{m \ge 0} \max_{n \ge m} \max\set{  \rltleval(\suff{u\emptyset^\omega}{n}, \varphi_2 )[3],  \max_{0 \le n' < n} \rltleval(\suff{u\emptyset^\omega}{n'}, \varphi_1)[3]  }\nonumber\\
=&\multicolumn{3}{l}{$\displaystyle\max\set{ \rltleval(\emptyset^\omega, \varphi_2 )[3], \max_{0 \le n' \le \size{u}} \rltleval(\suff{u\emptyset^\omega}{n'}, \varphi_1)[3]  }, $}
\end{align*}
where the equality again follows from all suffixes~$\suff{u\emptyset^\omega}{n}$ with $n \ge \size{u}$ being equal to $\emptyset^\omega$. Thus, we have derived the desired equality between $\rltleval(u\emptyset^\omega, \varphi)[2]$ and $\rltleval(u\emptyset^\omega, \varphi)[3]$.

\paragraph{Claim~\ref{lem:tech:3eq4})} The induction start and the induction steps for Boolean operators follow from the closure property, where we here pick $T$ to be the set of truth values from $\bools_4$ whose third and fourth bit coincide. For $u = u(0) \cdots u(\size{u}-1)$ and $n < \size{u}$, we define $\rot{u}{n} = u(n) \cdots u(\size{u}-1)u(0)\cdots u(n-1)$, \ie, $\rot{u}{n}$ is obtained by ``rotating'' $u$ $n$ times. The induction steps for the temporal operators are based on the following simple fact: The suffix~$\suff{u^\omega}{n}$ is equal to $(\rot{u}{n \bmod \size{u}})^\omega$, \ie, the induction hypothesis is applicable to the suffixes. In particular, $u^\omega$ has only finitely many distinct suffixes, which all appear infinitely often in a cyclic order.

Now, the induction steps for the next and until operator are analogous to their counterparts in Item~\ref{lem:tech:2eq3}, as the only property we require there is that the induction hypothesis is applicable to suffixes. Hence, due to Remark~\ref{rem:sugar}, it only remains to consider the inductive step for the release operator.

So consider a formula of the form~$\varphi = \varphi_1 \Rdot \varphi_2$. Then, we have, for an arbitrary~$u \in \Sigma^*$, that $\rltleval(u^\omega, \varphi)[3]$ is by definition equal to
\begin{align*}
&\min_{m \ge 0} \max_{n \ge m} \max\set{ \rltleval(\suff{u^\omega}{n}, \varphi_2 )[3], \max_{0 \le n' < n} \rltleval(\suff{u^\omega}{n'}, \varphi_1)[3]  }\\
=&\min_{m \ge 0} \max_{n \ge m} \max\set{ \rltleval(\suff{u^\omega}{n}, \varphi_2 )[4], \max_{0 \le n' < n} \rltleval(\suff{u^\omega}{n'}, \varphi_1)[4] }\\
=&\max_{0 \le n < \size{u}}\max\set{ \rltleval((\rot{u}{n})^\omega, \varphi_2 )[4], \max_{0 \le n' < n} \rltleval((\rot{u}{n'})^\omega, \varphi_1)[4]  },\end{align*}
where the first equality follows from twice applying the induction hypothesis and the second one is due to all suffixes~$\suff{u^\omega}{n}$ being equal to $\rot{u}{n \bmod \size{u}}^\omega$, and that there are only finitely many, which all appear infinitely often in a cyclic order among the $(\rot{u}{n})^\omega$ for $0 \le n < \size{u}$.

Similarly, $\rltleval(u^\omega, \varphi)[4]$ is by definition equal to 
\begin{align*}
&\max_{n \ge 0} \max\set{ \rltleval(\suff{u^\omega}{n}, \varphi_2 )[4], \max_{0 \le n' < n} \rltleval(\suff{u^\omega}{n'}, \varphi_1)[4]  }\\
=&\max_{0 \le n < \size{u}}\max\set{ \rltleval((\rot{u}{n})^\omega, \varphi_2 )[4], \max_{0 \le n' < n} \rltleval((\rot{u}{n'})^\omega, \varphi_1)[4]  },
\end{align*}
where the equality again follows from all suffixes~$\suff{u^\omega}{n}$ being equal to $\rot{u}{n \bmod \size{u}}^\omega$, and that there are only finitely many, which appear in a cyclic order: In particular, after the first $\size{u}$ suffixes, we have seen all of them. Thus, we have derived the desired equality between $\rltleval(u^\omega, \varphi)[3]$ and $\rltleval(u^\omega, \varphi)[4]$.

\paragraph{Claim~\ref{lem:tech:1eq2noreleaseonly})} The induction start and the induction steps for Boolean operators are covered by the closure property, where we here pick $T$ to be the set of truth values from $\bools_4$ whose first and second bit coincide. The cases of the next and until operator are again analogous to the first and second item. Hence, we only have to consider the inductive step for the always operator, as we here only consider formulas without release. 

So, consider a formula of the form~$\varphi = \Boxdot \varphi_1$. Here, we again rely on the fact that the suffix~$\suff{u^\omega}{n}$ is equal to $(\rot{u}{n \bmod \size{u}})^\omega$. By definition, $\rltleval(u^\omega, \varphi)[1]$ is equal to
\begin{align*}
\min_{n \ge 0} \rltleval(\suff{u^\omega}{n}, \varphi_1 )[1]
=\min_{n \ge 0} \rltleval(\suff{u^\omega}{n}, \varphi_1 )[2]
=\min_{0 \le n < \size{u}} \rltleval((\rot{u}{n})^\omega, \varphi_1 )[2],
\end{align*}
where the first equality is due to the induction hypothesis and the second one due to the fact that $u^\omega$ has only finitely many suffixes, which are all already realized by some $\suff{u^\omega}{n}$ for $0 \le n < \size{u}$.

Similarly, $\rltleval(u^\omega, \varphi)[2]$ is by definition equal to
\begin{align*}
\max_{m \ge 0}\min_{n \ge m} \rltleval(\suff{u^\omega}{n}, \varphi_1 )[1]
& = \max_{m \ge 0}\min_{n \ge m} \rltleval(\suff{u^\omega}{n}, \varphi_1 )[2] \\
& =\min_{0 \le n < \size{u}} \rltleval((\rot{u}{n})^\omega, \varphi_1 )[2],
\end{align*}
where the two equalities follow as before: the first by induction hypothesis and the second one by the fact that $u^\omega$ has only finitely many suffixes, which all appear infinitely often in a cyclic order and which are all already realized by some $\suff{u^\omega}{n}$ for $0 \le n < \size{u}$.
Thus, we have derived the desired equality between $\rltleval(u^\omega, \varphi)[1]$ and $\rltleval(u^\omega, \varphi)[2]$.
\qed\end{proof}

Now, we are able to prove Theorem~\ref{thm:realizabletruthvalues}.

\begin{proof}
We begin by showing that $0011$ and $0001$ are not realizable. 

First, towards a contradiction, assume there is an $\rltl$ formula~$\varphi$ and a prefix~$u$ such that $\rltlmoneval(u,\varphi) = 0011$, \ie, for every extension~$u\sigma$, we have $\rltleval(u\sigma)[2]=0$ and $\rltleval(u\sigma)[3]=1$. However, by picking $\sigma = \emptyset^\omega$ we obtain the desired contradiction to Lemma~\ref{lem:tech}.\ref{lem:tech:2eq3}. 

The proof for $0001$ is similar. Assume there is an $\rltl$ formula~$\varphi$ and a prefix~$u$ such that $\rltlmoneval(u,\varphi) = 0001$. Due to Lemma~\ref{lem:specificity}, we can assume that $u$ is non-empty. Thus, we  have $\rltleval(u^\omega, \varphi) = 0001$ by definition of $\rltlmoneval$, which contradicts Lemma~\ref{lem:tech}.\ref{lem:tech:3eq4}.

Finally, applying Lemma~\ref{lem:tech}.\ref{lem:tech:1eq2noreleaseonly}, one can show that no \rltl formula without the release operator realizes~$0111$. However, we show below that it is realizable by a formula with the release operator.

Next, we show that every other truth value~$\beta \notin \set{0011,0001}$ is  indeed realizable. The witnessing pairs of prefixes and formulas are presented in Table~\ref{tag:truthvalues}. 

First, consider $\beta = 0111$ with prefix~$u = \emptyset\set{a}$ and formula~$\varphi = a \Rdot a$. We have $\rltlsem(1, \varphi) = a \R a$ and $\rltlsem(2,\varphi) = \Diamond\Box a \vee \Diamond a$. Note that $a \R a$ is violated by $u \sigma$, for every $\sigma \in \Sigma^\omega$. Dually, $\Diamond\Box a \vee \Diamond a$ is satisfied by $u \sigma$, for every $\sigma \in \Sigma^\omega$. Hence, for arbitrary~$\sigma \in \Sigma^\omega$, we have $\rltleval(u\sigma,\varphi)[1] = 0$ and $\rltleval(u\sigma,\varphi)[2] = 1$. Hence, we have $\rltleval(u\sigma,\varphi) = 0111$ for every $\sigma$, as this is the only truth value that matches this pattern. Hence, by definition, we obtain $\rltlmoneval(u, \varphi) = 0111$. 

The verification for all other truth values is based on Remark~\ref{rem:intuitivesemantics}, which is applicable to all formulas~$\varphi$ in the third column witnessing the realization of a truth value~$\beta \neq 0111$.
Now, for every such truth value~$\beta$ and corresponding pair~$(u, \varphi)$, one can easily verify the following:
\begin{itemize}
	\item If $\beta[i] = 0$, then no $u \sigma$ satisfies $\rltlsem(i, \varphi)$. 
	\item If $\beta[i] = 1$, then every $u \sigma$ satisfies $\rltlsem(i, \varphi)$.
	\item If $\beta[i] = ?$, then there are $\sigma,\sigma'$ such that $u \sigma$ satisfies $\rltlsem(i, \varphi)$ and such that $u \sigma'$ violates $\rltlsem(i, \varphi)$. In all such cases, $\sigma,\sigma' \in \set{\emptyset^\omega, \set{a}^\omega, \set{a}\emptyset^\omega, \emptyset\set{a}^\omega, (\set{a}\emptyset)^\omega}$ suffice.
\end{itemize}
We leave the details of this slightly tedious, but trivial, verification to the reader.
\qed\end{proof}

As shown in Table~\ref{tag:truthvalues}, all of the realizable truth values except for $0111$ are realized by formulas using only conjunction, disjunction, negation, eventually, and always. Further, $0111$ can only be realized by a formula with the release operator while the truth values $0011$ and $0001$ are indeed not realizable at all.

Note that the two unrealizable truth values $0011$ and $0001$ both contain a $0$ that is directly followed by a $1$. The proof of unrealizability formalizes the intuition that such an ``abrupt'' transition from definitive violation of a property to definitive satisfaction of the property cannot be witnessed by any finite prefix. Finally, the only other truth value of this form, $0111$, is only realizable by using a formula with the release operator.

Going again back to the motivating example~$\Boxdot s$, consider the evolution of the truth values on the sequence~$\varepsilon, \set{s}, \set{s}\emptyset$: They are $\ukno \ukno \ukno \ukno $, $\ukno \ukno \ukno 1$, and $0 \ukno \ukno  1$, \ie, $0$'s and $1$'s are stable when extending a prefix, only a $\ukno$ may be replaced by a $0$ or a $1$. 
This property holds in general. To formalize this, say that $\beta' \in \ternaries_4$ is more specific than $\beta \in \ternaries_4$, written as $\beta \preceq \beta'$, if, for all $i$, $\beta[i] \neq \ukno$ implies $\beta'[i] = \beta[i]$. 

\begin{lemma}
\label{lem:specificity}
Let $\varphi  $ be an \rltl formula and $u,u' \in \Sigma^*$. If $u \sqsubseteq u'$, then $\rltlmoneval (u, \varphi) \preceq \rltlmoneval (u', \varphi)$.
\end{lemma}

\begin{proof}
Let $u \sqsubseteq u'$ and assume we have $\rltlmoneval(u,\varphi)[i] \in \set{0,1}$. Thus, by definition, $\rltleval(u\sigma, \varphi)[i] = \rltlmoneval(u,\varphi)[i]$ for every $\sigma \in \Sigma^\omega$. Now, as $u$ is a prefix of $u'$, we can decompose $u'$ into $u' = u v$ for some $v \in\Sigma^*$ and every extension~$u' \sigma'$ of $u'$ is the extension~$u v \sigma'$ of $u$. Hence, we have $\rltleval(u'\sigma', \varphi)[i] = \rltleval(uv\sigma', \varphi)[i] = \rltlmoneval(u,\varphi)[i]$ for every $\sigma' \in \Sigma^\omega$. Thus, $\rltlmoneval(u',\varphi)[i] = \rltlmoneval(u,\varphi)[i] $. 

As this property holds for every~$i$, we obtain $\rltlmoneval (u, \varphi) \preceq \rltlmoneval (u', \varphi)$.
\qed\end{proof}

Let us discuss two properties of the semantics: \emph{impartiality} and \emph{anticipation}~\cite{DBLP:conf/rv/DeckerLT13}.
Impartiality states that a definitive verdict will never be revoked: If $\rltlmoneval (u, \varphi) [i] \neq \ukno$, then for all finite extensions $v \in \Sigma^*$, the verdict will not change, so $\rltlmoneval (uv, \varphi)[i] = \rltlmoneval (u, \varphi) [i]$. 
This property follows immediately from Lemma~\ref{lem:specificity}.
Anticipation requires that a definitive verdict is decided as soon as possible, \ie, if $\rltlmoneval (u, \varphi) [i] = \ukno$, then $u$ can still be extended to satisfy and to violate $\varphi$ with the $i$-th bit.
Formally, there have to exist infinite extensions~$\sigma_0$ and $\sigma_1$ such that $\rltleval(u\sigma_0, \varphi)[i] =0$ and $\rltleval(u\sigma_1, \varphi)[i] =1$.
Anticipation holds by  definition of $\rltlmoneval (u, \varphi)$.
 
Due to Lemma~\ref{lem:specificity}, for a fixed formula, the prefixes of every infinite word can assume at most five different truth values, which are all of increasing specificity. It turns out that this upper bound is tight. To formalize this claim, we denote the strict version of $\preceq$ by $\prec$, \ie, $\beta \prec \beta'$ if and only if $\beta \preceq \beta'$ and $\beta \neq \beta'$. 

\begin{lemma}
\label{lem:specificity:lb}
There is an \rltl~formula~$\varphi$ and prefixes~$u_0 \sqsubset u_1 \sqsubset u_2 \sqsubset u_3 \sqsubset u_4 $ such that 
$ \rltlmoneval (u_0, \varphi) \prec \rltlmoneval (u_1, \varphi) \prec \rltlmoneval (u_2, \varphi) \prec \rltlmoneval (u_3, \varphi) \prec \rltlmoneval (u_4, \varphi)$.
\end{lemma}

\begin{proof}
Consider the sequence~$\beta_0, \ldots, \beta_4$ with $\beta_j = 0^j \ukno^{4-j}$ and note that we have $\beta_j \prec \beta_{j+1}$ for every $j<4$. Furthermore, let $u_j = \emptyset^j$ for $j \in \set{0, \ldots, 4}$. We construct a formula~$\varphi$ such that $\rltlmoneval(u_j, \varphi) = \beta_j$ for every $j \in \set{0, \ldots, 4}$.

To this end, let
\begin{itemize}
	\item $\psi_{\beta_1} = \Diamonddot(a \wedge \Boxdot\neg \Diamonddot a) $,
	
	\item $\psi_{\beta_2} = \Boxdot(a \wedge \Xdot \neg a) \wedge \neg\Diamonddot\neg \Diamonddot a $, and 
	
	\item $\psi_{\beta_3} = \Diamonddot \Boxdot a \wedge \Diamonddot\neg \Diamonddot a $.
\end{itemize}
Later, we rely on the following fact about these formulas, which can easily be shown by applying Remark~\ref{rem:intuitivesemantics}: We have $\rltlmoneval(u, \psi_{\beta_j}) = \beta_j$ for every prefix~$u$.

Further, for $j \in \set{0,1,2,3}$, let $\psi_j$ be a formula that requires the proposition~$a$ to be violated at the first $j - 1$ positions, but to hold at the $j$-th position (recall that we start counting at zero), \ie,
$
\psi_j = (\bigwedge_{0 \le j' < j} \Xdot^{j'} \neg a) \wedge \Xdot^j a
$.
Here, we define the nesting of next operators as usual: $\Xdot^0 \xi = \xi$ and $\Xdot^{j+1} \xi= \Xdot\Xdot^j\xi$. By definition, we have $\rltleval(\emptyset^{j+1}\sigma, \psi_j) = 0000$ for every $\sigma\in\Sigma^\omega$ ($\dagger$). 

Now, we define
\[
\varphi = \psi_0 \vee \bigvee_{j=1}^3 \left(\psi_{\beta_j} \wedge \psi_j\right)
\]
and claim that it has the desired properties. To this end, we note that property~($\dagger$) implies $\rltleval(\emptyset^4\sigma, \varphi) = 0000$ for every $\sigma\in\Sigma^\omega$ ($\dagger\dagger$), as every disjunct of $\varphi$ contains a conjunct of the form~$\psi_j$ for some $j \le 3$. Also, let us mention that Remark~\ref{rem:intuitivesemantics} is applicable to $\varphi$. 

It remains to prove $\rltlmoneval(u_j, \varphi) = \beta_j$ for every $j \in \set{0, \ldots ,4}$.
\begin{itemize}
	
	\item For $j = 0$, we have $u_0 = \varepsilon$ and $\beta_0 = \ukno \ukno \ukno \ukno $. Hence, it suffices to present $\sigma_0, \sigma_1 \in \Sigma^\omega$ such that $\rltleval(\sigma_0, \varphi) = 0000$ and $\rltleval(\sigma_1, \varphi) = 1111$. 
	
	Due to property~($\dagger\dagger$), we can pick $\sigma_0 = \emptyset^\omega$. To conclude, we pick $\sigma_1 = \set{a}^\omega$, as we have 
	\[\rltleval(\sigma_1, \varphi) \ge \rltleval(\sigma_1, \psi_0) = \rltleval(\set{a}^\omega, a) = 1111 ,\] where the first inequality follows from $\psi_0$ being a disjunct of $\varphi$. 
	
	\item For $j = 1$, we have $u_1 = \emptyset$ and $\beta_1 = 0 \ukno \ukno \ukno$. To show $\rltlmoneval(u_1, \varphi) = \beta_1$, it suffices to present $\sigma_0,\sigma_1 \in\Sigma^\omega$ such that $\rltleval(u_1\sigma_0, \varphi) = 0000$, $\rltleval(u_1\sigma_1, \varphi) = 0111$, and show that $\rltleval(u_1\sigma, \varphi)[1] = 0$ for every $\sigma\in\Sigma^\omega$. First, we again pick~$\sigma_0 = \emptyset^\omega$ due to property~($\dagger\dagger$). 
	Now, consider $\sigma_1 = \set{a}^\omega$. Then,
	\begin{align*}
		\rltleval(u_1 \set{a}^\omega , \psi_{\beta_j} \wedge \psi_j)  
		& = \min\set{ \rltleval(u_1 \set{a}^\omega , \psi_{\beta_j} ), \rltleval(u_1 \set{a}^\omega ,\psi_j) } \\
		& = \min\set{0111, 1111}= 0111,
	\end{align*}
	where $\rltleval(u_1 \set{a}^\omega , \psi_{\beta_j} ) = 0111$ can easily be verified using Remark~\ref{rem:intuitivesemantics}. 
To conclude, using Remark~\ref{rem:intuitivesemantics}, one can easily verify that $\rltlsem(1, \varphi)$ is not satisfied by $u_1\sigma$ for any $\sigma\in\Sigma^\omega$.
	
\item The reasoning for $j =2,3$ is along the same lines as the one for $j=1$ and is left to the reader.	
	
	\item For $j = 4$, we have $u_4 = \emptyset\emptyset\emptyset\emptyset$ and $\beta_4 = 0000$. Hence, our claim follows directly from property~($\dagger\dagger$), which shows $\rltleval(u_4\sigma, \varphi) = 0000$ for every $\sigma \in \Sigma^\omega$.\hfill\qed
\end{itemize}
\end{proof}

After determining how many different truth values can be assumed by prefixes of a single infinite word, an obvious question is how many truth values can be realized by a fixed formula on \emph{different} prefixes. It is not hard to combine the formulas in Table~\ref{tag:truthvalues} to a formula that realizes all truth values not ruled out by Theorem~\ref{thm:realizabletruthvalues}.\footnote{\label{bingchenfootnote}Note that there are formulas in publicly available repositories that assume \emph{many} truth values. One example is the formula
\[
(((a \wedge d) \vee (\neg a \wedge \neg d)) \wedge \Boxdot (\neg b \vee (\neg a \wedge d))) \vee  (((\neg a \wedge d) \vee (a \wedge \neg d)) \wedge \Diamonddot (b \wedge (a \vee \neg d))) \vee (a \wedge \Boxdot b),
\]
which is taken from the LTLStore~\cite{DBLP:journals/corr/abs-1807-03296} and assumes ten different truth values.
 }

\begin{lemma}
\label{lem:oneformulaalltruthvalues}
There is an \rltl formula~$\varphi$ 	such that for every $\beta \in \ternaries_4 \setminus \set{0011,0001}$ there is a prefix~$u_\beta$ with $\rltlmoneval(u_\beta, \varphi) = \beta$. 
\end{lemma}

\begin{proof}
For every $\beta \in \ternaries_4 \setminus \set{0011,0001}$ let $\varphi_\beta$ be an \rltl formula and $u_\beta'$ be a prefix, both over~$\set{a}$, with $\rltlmoneval(u_\beta', \varphi_\beta) = \beta$. 
Such formulas and prefixes exist as shown in Table~\ref{tag:truthvalues}.

Now, consider the formula
\[
\varphi = \bigvee_{ \beta \in \ternaries_4 \setminus \set{0011,0001}} a_\beta \wedge \varphi_\beta
\]
over the propositions $\set{a} \cup \set{a_\beta \mid \beta \in \ternaries_4 \setminus \set{0011,0001}}$.

By construction, we have $\rltlmoneval(u_\beta, \varphi) = \beta$ for every $\beta$, where 
\[u_\beta = (u_\beta'(0) \cup \set{a_\beta}) u_\beta'(1) \cdots u_\beta'(\size{u_\beta'}-1),\] i.e, we obtain $u_\beta$ from $u_\beta'$ by adding the proposition~$a_\beta$ to the first letter.
Hence, $\varphi$ has the desired properties.
\qed\end{proof}

Finally, let us consider the notion of \emph{monitorability}~\cite{DBLP:conf/fm/PnueliZ06}, an important concept in the theory of runtime monitoring. As a motivation, consider the  \ltl formula~$\psi = \Box\Diamond s$ and an arbitrary prefix~$u \in \Sigma^*$. Then, the extension~$u \set{s}^\omega$ satisfies $\psi$ while the extension~$u \emptyset^\omega$ does not satisfy $\psi$, \ie, satisfaction of $\psi$ is independent of any prefix~$u$. Hence, we have $\ltlmoneval (u, \psi) = \ukno$ for every prefix~$u$, \ie, monitoring the formula~$\psi$ does not generate any information. 

In general, for a fixed \ltl formula~$\varphi$, a prefix~$u \in \Sigma^*$ is called \emph{ugly} if we have $\ltlmoneval(uv ,\varphi) = \ukno$ for every finite~$v\in\Sigma^*$, \ie, every finite extension of $u$ yields an indefinite verdict.%
\footnote{Note that the good/bad prefixes introduced by Kupfermann and Vardi~\cite{DBLP:journals/fmsd/KupfermanV01} can only be extended into infinite words satisfying/unsatisfying the formula, respectively, and thus provide a verdict immediately. On the other hand, no finite extension of an ugly prefix~\cite{BauerLeuckerSchallhart11} allows to conclude on the satisfaction of the formula.}
Now, $\varphi$ is \emph{\ltl-monitorable} if there is no ugly prefix with respect to $\varphi$. 
A wide range of  \ltl formulas (\eg, $\psi = \Box\Diamond s$ as above) are unmonitorable in that sense. In particular, 44\% of the \ltl formulas considered in the experiments of Bauer \etal are not \ltl-monitorable. 

We next generalize the notion of monitorability to \rltl. 
In particular, we answer whether there are unmonitorable \rltl formulas. 
Then, in  Section~\ref{sec:evaluation}, we exhibit that all \ltl formulas considered by Bauer \etal's experimental evaluation, even the unmonitorable ones, are monitorable under \rltl semantics. 
To conclude the motivating example,  note that the \rltl analogue~$\Boxdot\Diamonddot s$ of the \ltl formula~$\psi$ induces two truth values from $\ternaries_4$ indicating whether $s$ has been true at least once (truth value~$\ukno\ukno\ukno1$) or not (truth value~$\ukno\ukno\ukno\ukno$).
Even more so, every prefix inducing the truth value~$\ukno \ukno \ukno \ukno $ can be extended to one inducing the truth value~$\ukno \ukno \ukno 1$.

\begin{definition}
Let $\varphi$ be an \rltl formula.
 A prefix~$u \in \Sigma^*$ is called \emph{ugly} if we have $\ltlmoneval(uv ,\varphi) = \ukno\ukno\ukno\ukno$ for every finite~$v\in\Sigma^*$.
Further, $\varphi$ is \emph{\rltl-monitorable} if it has no ugly prefix. 
\end{definition}

As we have argued above, the formula~$\Boxdot\Diamonddot s$ has no ugly prefix, \ie, it is \rltl-monitorable. Thus, we have found an unmonitorable \ltl formula whose \rltl analogue (the formula obtained by adding dots to all temporal operators) is monitorable. 
The converse statement is also true. There is a monitorable \ltl formula whose \rltl analogue is unmonitorable. To this end, consider the \ltl formula
\[
(\Box s \wedge \Box\neg s) \Rimplies (\Diamond \Box s \wedge \Diamond \neg \Diamond s),
\]
which is a tautology and therefore monitorable. On the other hand, we claim that $\emptyset\set{s}$ is an ugly prefix for the \rltl analogue~$\varphi$ obtained by adding dots to the temporal operators. 
To this end note that we have both $\rltleval(\emptyset\set{s} v \emptyset^\omega, \varphi) = 1111$ and $\rltleval(\emptyset\set{s} v \set{s}^\omega, \varphi) = 0000$ for every~$v \in \Sigma^*$. 
Hence, $\rltlmoneval(\emptyset\set{s}v,\varphi) = \ukno \ukno \ukno \ukno $ for every such $v$, \ie, $\emptyset\set{s}$ is indeed ugly and $\varphi$ therefore not \rltl-monitorable. 

Thus, there are formulas that are unmonitorable under \ltl semantics, but monitorable under \rltl semantics and there are formulas that are unmonitorable under \rltl semantics, but monitorable under \ltl semantics.
Using these formulas one can also construct a formula that is unmonitorable under both semantics. 

To this end, fix \ltl formulas~$\varphi_\ell$ and $\varphi_r$ over disjoint sets of propositions and a fresh proposition~$p$ not used in either formula such that 
\begin{itemize}
\item $\varphi_\ell$ has an ugly prefix~$u_\ell$ under \ltl semantics, and
\item $\varphi_r$ (with dotted operators) has an ugly prefix~$u_r$ under \rltl semantics.
\end{itemize}
We can assume both prefixes to be non-empty, as ugliness is closed under finite extensions.
Let $\varphi = (p \wedge \varphi_\ell) \vee (\neg p \wedge \varphi_r)$.
Then, the prefix obtained from $u_\ell$ by adding the proposition~$p$ to the first letter is ugly for $\varphi$ under \ltl semantics and $u_r$ is ugly for $\varphi$ (with dotted operators) under \rltl semantics.

As a final example, recall that we have shown that $\Boxdot \Diamonddot s$ is \rltl-monitorable and consider its negation~$\neg \Boxdot \Diamonddot s$. 
It is not hard to see that $\rltlmoneval(u, \varphi) = \ukno \ukno \ukno \ukno $ holds for every prefix~$u$.
Hence, $\varepsilon$ is an ugly prefix for the formula, \ie, we have found another unmonitorable \rltl formula. 
In particular, the example shows that, unlike for \ltl, \rltl-monitorability is not preserved under negation.

After having studied properties of rLTL monitorability, we next show our main result: The robust monitoring semantics~$\rltlmoneval$ can be implemented by finite-state machines.

\section{Construction of \rltl Monitors}
\label{sec:monitoring}

An \rltl monitor is an implementation of the robust monitoring semantics $\rltlmoneval $ in form of a finite-state machine with output.
More precisely, an \emph{\rltl monitor} for an \rltl formula~$\varphi$ is a finite-state machine $\mach_\varphi$ that on reading an input $u \in \Sigma^*$ outputs $\rltlmoneval (u, \varphi)$.
In this section, we show how to construct \rltl monitors and that this construction is asymptotically not more expensive than the construction of \ltl monitors.
Let us fix an \rltl formula~$\varphi$ for the remainder of this section.

Our \rltl monitor construction is inspired by Bauer et al.~\cite{BauerLeuckerSchallhart11} and generates a sequence of finite-state machines (\ie, Büchi automata over infinite words, (non)deterministic automata over finite words, and Moore machines).
Underlying these machines are \emph{transition structures} $\tsys = (Q, q_\init, \Delta)$ consisting of a nonempty, finite set $Q$ of states, an initial state $q_\init \in Q$, and a transition relation $\Delta \subseteq Q \times \Sigma \times Q$.
An (infinite) run of~$\tsys$ on a word $\sigma = a_0 a_1 a_2 \cdots \in \Sigma^\omega$ is a sequence $\rho = q_0 q_1 \cdots$ of states such that $q_0 = q_\init$ and $(q_{j}, a_{j}, q_{j+1}) \in \Delta$ for $j \in \mathbb N$.
Finite runs on finite words are defined analogously.
The transition structure~$\tsys$ is \emph{deterministic} if
\begin{enumerate*}[label={(\alph*)}]
	\item $(q, a, q') \in \Delta$ and $(q, a, q'') \in \Delta$ imply $q' = q''$ and
	\item for each $q \in Q$ and $a \in \Sigma$ there exists a $(q, a, q') \in \Delta$.
\end{enumerate*}
We then replace the transition relation~$\Delta$ by a function $\delta\colon Q \times \Sigma \rightarrow Q$.
Finally, we define the \emph{size} of a transition structure $\tsys$ as $\card{\tsys} = \card{Q}$ in order to measure its complexity.

Our construction then proceeds in three steps:
\begin{enumerate}
	\item We bring $\varphi$ into an operational form by constructing Büchi automata $\aut_\beta^\varphi$ for each truth value $\beta \in \bools_4$ that can decide the valuation $\rltleval(\sigma, \varphi)$ of infinite words $\sigma \in \Sigma^\omega$.
	\item Based on these Büchi automata, we then construct nondeterministic automata $\autB_\beta^\varphi$ that can decide whether a finite word $u \in \Sigma^*$ can still be extended to an infinite word $u\sigma \in \Sigma^\omega$ with $\rltleval(u\sigma, \varphi) = \beta$.
	\item We determinize the nondeterministic automata obtained in Step~2 and combine them into a single Moore machine that computes $\rltlmoneval (u, \varphi)$.
\end{enumerate}
Let us now describe each of these steps in detail.

\paragraph{\bfseries Step~1:}
We first translate the \rltl formula~$\varphi$ into several Büchi automata using a construction by Tabuada and Neider~\cite{DBLP:conf/csl/TabuadaN16}, summarized in Theorem~\ref{thm:büchi} below.
A \emph{(nondeterministic) Büchi automaton (NBA)} is a four-tuple $\aut = (Q, q_\init, \Delta, F)$ where $\tsys = (Q, q_\init, \Delta)$ is a transition structure and $F \subseteq Q$ is a set of accepting states.
A run~$\pi$ of~$\aut$ on~$\sigma \in \Sigma^\omega$ is a run of~$\tsys$ on~$\sigma$, and we say that~$\pi$ is accepting if it contains infinitely many states from~$F$.
The automaton~$\aut$ accepts a word~$\sigma$ if there exists an accepting run of~$\aut$ on~$\sigma$.
The language $\lang(\aut)$ is the set of all words accepted by~$\aut$, and the size of~$\aut$ is defined as $\card{\aut} = \card{\tsys}$.

\begin{theorem}[Tabuada and Neider~\cite{DBLP:conf/csl/TabuadaN16}] \label{thm:büchi}
Given a truth value $\beta \in \bools_4$, one can construct a Büchi automaton $\aut_\beta^\varphi$ with $2^{\bigo(|\varphi|)}$ states such that $\lang(\aut_\beta^\varphi) = \{ \sigma \in \Sigma^\omega \mid \rltleval(\sigma, \varphi) = \beta \}$.
This construction can be performed in $2^{\bigo(|\varphi|)}$ time.
\end{theorem}

The Büchi automata~$\mathcal A_\beta^\varphi$ for $\beta \in \bools_4$ serve as building blocks for the next~steps.
However, before we proceed, let us illustrate this step with an example.

\begin{example} \label{ex:construction-step-1}
Let us consider the formula $\varphi = \Boxdot s$, which already served as a running example in Section~\ref{sec:problem}.
Applying Theorem~\ref{thm:büchi} results in the five nondeterministic Büchi automata~$\aut_\beta^\varphi$, one for each $\beta \in \bools_4$, shown in Figure~\ref{fig:ex-step-1}.
We here use the standard way to represent finite-state machines graphically.
States are drawn as circles and transitions are drawn as arrows.
Moreover, the initial state has an incoming arrow, while accepting states are indicted by double circles.
Finally, note that we use propositional formulas to symbolically define sets of transitions.
For instance, a transition labeled with $s$ in Figure~\ref{fig:ex-step-1:1111} represents all transitions labeled with a symbol from the set $\{ A \subseteq P \mid s \in A \} \subseteq \Sigma$.
In particular, $\mathit{true}$ represents all symbols in $\Sigma$. \hfill\exampleend
\end{example}

\begin{figure}
	\centering
	\begin{subfigure}[b]{.25\textwidth}
		\centering
		\begin{tikzpicture}[auto]
			\node[state, accepting] (0) at (0, 0) {};
			\draw[<-, shorten <=1pt] (0.west) -- +(-.3,0);
			\path[->] (0) edge[loop above] node {$s$} ();
		\end{tikzpicture}
		\caption{The NBA $\aut^{\Boxdot s}_{1111}$}
		\label{fig:ex-step-1:1111}
	\end{subfigure}
	\hskip 2em
	\begin{subfigure}[b]{.3\textwidth}
		\centering
		\begin{tikzpicture}[auto]
			\node[state] (0) at (0, 0) {};
			\node[state, accepting] (1) at (2.25, 0) {};
			\draw[<-, shorten <=1pt] (0.west) -- +(-.3,0);
			\path[->] (0) edge[loop above] node {$\mathit{true}$} () edge node {$\lnot s$} (1) ;
			\path[->] (1) edge[loop above] node {$s$} ();
		\end{tikzpicture}
		\caption{The NBA $\aut^{\Boxdot s}_{0111}$}
		\label{fig:ex-step-1:0111}
	\end{subfigure}
	\hskip 2em
	\begin{subfigure}[b]{.3\textwidth}
		\centering
		\begin{tikzpicture}[auto]
			\node[state] (0) at (0, 0) {};
			\node[state] (1) at (2.25, 0) {};
			\node[state, accepting] (2) at (1.125, 1.75) {};
			\draw[<-, shorten <=1pt] (0.west) -- +(-.3,0);
			\path[->] (0) edge[loop above] node {$\lnot s$} () edge node[swap] {$s$} (1) ;
			\path[->] (1) edge node[swap, near end] {$\lnot s$} (2) edge[loop above] node {$s$} ();
			\path[->] (2) edge node {$\lnot s$} (0) edge[bend right=20] node[swap, near end] {$s$} (1);
		\end{tikzpicture}
		\caption{The NBA $\aut^{\Boxdot s}_{0011}$}
		\label{fig:ex-step-1:0011}
	\end{subfigure}

	\bigskip
	\begin{subfigure}[b]{.3\textwidth}
		\centering
		\begin{tikzpicture}[auto]
			\node[state] (0) at (0, 0) {};
			\node[state, accepting] (1) at (2.25, 0) {};
			\draw[<-, shorten <=1pt] (0.west) -- +(-.3,0);
			\path[->] (0) edge[loop above] node {$\mathit{true}$} () edge node {$s$} (1) ;
			\path[->] (1) edge[loop above] node {$\lnot s$} ();
			
		\end{tikzpicture}
		\caption{The NBA $\aut^{\Boxdot s}_{0001}$}
		\label{fig:ex-step-1:0001}
	\end{subfigure}
	\hskip 2em
	\begin{subfigure}[b]{.25\textwidth}
		\centering
		\begin{tikzpicture}[auto]
			\node[state, accepting] (0) at (0, 0) {};
			\draw[<-, shorten <=1pt] (0.west) -- +(-.3,0);
			\path[->] (0) edge[loop above] node {$\lnot s$} ();
		\end{tikzpicture}
		\caption{The NBA $\aut^{\Boxdot s}_{0000}$}
		\label{fig:ex-step-1:0000}
	\end{subfigure}
	
	\caption{The Büchi automata $\aut^{\Boxdot s}_\beta$ constructed in Step~1 of our monitor construction}
	\label{fig:ex-step-1}
\end{figure}

\paragraph{\bfseries Step 2:}
For each Büchi automaton $\aut^\varphi_\beta$ obtained in the previous step, we now construct a nondeterministic automaton $\autB_\beta^\varphi$ over finite words.
This automaton determines whether a finite word $u \in \Sigma^*$ can be continued to an infinite word $u\sigma \in \lang(\aut_\beta^\varphi)$ (\ie, $\rltleval(u\sigma, \varphi) = \beta$) and is used later to construct the \rltl monitor.

A \emph{nondeterministic finite automaton (NFA)} is a four-tuple $\aut = (Q, q_\init, \Delta, F)$ that is syntactically identical to a Büchi automaton.
The size of~$\aut$ is defined analogously to Büchi automata.
In contrast to Büchi automata, however, NFAs only admit finite runs on finite words, \ie, a run of $\aut $ on $u = a_0 \cdots a_{n-1} \in \Sigma^*$ is a sequence~$q_0 \cdots q_n$ such that $q_0 = q_\init$ and $(q_j, a_j, q_{j+1}) \in \Delta$ for every $j < n$.
A run $q_0 \cdots q_n$ is called \emph{accepting} if $q_n \in F$.
Accepted words as well as the language of~$\aut$ are again defined analogously to the Büchi case.
If~$(Q, q_\init, \Delta)$ is deterministic,~$\aut$ is a \emph{deterministic finite automaton (DFA)}.
It is well-known that for each NFA~$\aut$ one can construct a DFA~$\aut'$ with~$\lang(\aut) = \lang(\aut')$ and $\card{\aut'} \in \bigo(2^{\card{\aut}})$.

Given the Büchi automaton $\aut_\beta^\varphi = (Q_\beta, q_{\init,\beta}, \Delta_\beta, F_\beta)$, we first compute the set $F_\beta^\star = \{ q \in Q_\beta \mid \lang(\aut^\varphi_\beta(q)) \neq \emptyset \}$, where $\aut^\varphi_\beta(q)$ denotes the Büchi automaton $\aut^\varphi_\beta$ but with initial state $q$ instead of $q_\init$.
Intuitively, the set $F_\beta^\star$ contains all states $q \in Q_\beta$ from which there exists an accepting run in $\aut^\varphi_\beta$ and, hence, indicates whether a finite word $u \in \Sigma^*$ reaching a state of $F_\beta^\star$ can be extended to  an infinite word $u\sigma' \in \lang(\aut_\beta^\varphi)$.
The set $F_\beta^\star$ can be computed, for instance, using a nested depth-first search~\cite{DBLP:conf/tacas/SchwoonE05} for each state $q \in Q_\beta$.
Since each such search requires time quadratic in $\card{\aut^\varphi_\beta}$, the set $F_\beta^\star$ can be computed in time $\bigo(\card{\aut^\varphi_\beta}^3)$.

Using $F_\beta^\star$, we define the NFA $\autB^\varphi_\beta = (Q_\beta, q_{\init,\beta}, \Delta_\beta, F_\beta^\star)$. It shares the transition structure of $\aut_\beta^\varphi$ and uses $F_\beta^\star$ as the set of accepting states. Let us illustrate this construction using our running example.

\begin{example} \label{ex:construction-step-2}
Given the NBAs $\aut^\varphi_\beta$ from Step~1 of our construction, we now compute the corresponding NFAs $\autB^\varphi_\beta$, which are depicted in Figure~\ref{fig:ex-step-2}.
Note that the transition structure has remained the same as compared to the preceding step (see Figure~\ref{fig:ex-step-1}).
By contrast, the accepting states have changed according to the definition of $F^\star_\beta$, causing all states to be accepting.
Note, however, that this does not mean that the resulting NFAs accept any finite word.
For instance, the NFA~$\autB^{\Boxdot s}_{1111}$ in Figure~\ref{fig:ex-step-2:1111} is a counterexample to this claim. \hfill\exampleend
\end{example}

\begin{figure}
	\centering
	\begin{subfigure}[b]{.25\textwidth}
		\centering
		\begin{tikzpicture}[auto]
			\node[state, accepting] (0) at (0, 0) {};
			\draw[<-, shorten <=1pt] (0.west) -- +(-.3,0);
			\path[->] (0) edge[loop above] node {$s$} ();
		\end{tikzpicture}
		\caption{The NFA $\autB^{\Boxdot s}_{1111}$}
		\label{fig:ex-step-2:1111}
	\end{subfigure}
	\hskip 2em
	\begin{subfigure}[b]{.3\textwidth}
		\centering
		\begin{tikzpicture}[auto]
			\node[state, accepting] (0) at (0, 0) {};
			\node[state, accepting] (1) at (2.25, 0) {};
			\draw[<-, shorten <=1pt] (0.west) -- +(-.3,0);
			\path[->] (0) edge[loop above] node {$\mathit{true}$} () edge node {$\lnot s$} (1) ;
			\path[->] (1) edge[loop above] node {$s$} ();
		\end{tikzpicture}
		\caption{The NFA $\autB^{\Boxdot s}_{0111}$}
		\label{fig:ex-step-2:0111}
	\end{subfigure}
	\hskip 2em
	\begin{subfigure}[b]{.3\textwidth}
		\centering
		\begin{tikzpicture}[auto]
			\node[state, accepting] (0) at (0, 0) {};
			\node[state, accepting] (1) at (2.25, 0) {};
			\node[state, accepting] (2) at (1.125, 1.75) {};
			\draw[<-, shorten <=1pt] (0.west) -- +(-.3,0);
			\path[->] (0) edge[loop above] node {$\lnot s$} () edge node[swap] {$s$} (1) ;
			\path[->] (1) edge node[swap, near end] {$\lnot s$} (2) edge[loop above] node {$s$} ();
			\path[->] (2) edge node {$\lnot s$} (0) edge[bend right=20] node[swap, near end] {$s$} (1);
		\end{tikzpicture}
		\caption{The NFA $\autB^{\Boxdot s}_{0011}$}
		\label{fig:ex-step-2:0011}
	\end{subfigure}

	\bigskip
	\begin{subfigure}[b]{.3\textwidth}
		\centering
		\begin{tikzpicture}[auto]
			\node[state, accepting] (0) at (0, 0) {};
			\node[state, accepting] (1) at (2.25, 0) {};
			\draw[<-, shorten <=1pt] (0.west) -- +(-.3,0);
			\path[->] (0) edge[loop above] node {$\mathit{true}$} () edge node {$s$} (1) ;
			\path[->] (1) edge[loop above] node {$\lnot s$} ();
		\end{tikzpicture}
		\caption{The NFA $\autB^{\Boxdot s}_{0001}$}
		\label{fig:ex-step-2:0001}
	\end{subfigure}
	\hskip 2em
	\begin{subfigure}[b]{.25\textwidth}
		\centering
		\begin{tikzpicture}[auto]
			\node[state, accepting] (0) at (0, 0) {};
			\draw[<-, shorten <=1pt] (0.west) -- +(-.3,0);
			\path[->] (0) edge[loop above] node {$\lnot s$} ();
		\end{tikzpicture}
		\caption{The NFA $\autB^{\Boxdot s}_{0000}$}
		\label{fig:ex-step-2:0000}
	\end{subfigure}
	
	\caption{The NFAs $\autB^{\Boxdot s}_\beta$ constructed in Step~2 of our monitor construction}
	\label{fig:ex-step-2}
\end{figure}

The next lemma now states that $\autB_\beta^\varphi$ indeed recognizes prefixes of words in~$\lang(\aut_\beta^\varphi)$.

\begin{lemma} \label{lem:B:correct}
Let~$\beta \in \bools_4$ and $u \in \Sigma^*$.
Then, $u \in \lang(\autB_\beta^\varphi)$ if and only if there exists an infinite word $\sigma \in \Sigma^\omega$ with $\rltleval(u\sigma, \varphi) = \beta$.
\end{lemma}

\begin{proof}
We show both directions separately.

\paragraph{From left to right:}
Assume $u \in \lang(\autB_\beta^\varphi)$.
Moreover, let $q \in F^\star_\beta$ be the accepting state reached by $\autB_\beta^\varphi$ on an accepting run on $u$ (which exists since $u \in \lang(\autB^\varphi_\beta)$).
By definition of $F^\star_\beta$, this means that $\lang(\aut^\varphi_\beta(q)) \neq \emptyset$, say $\sigma \in \lang(\aut^\varphi_\beta(q))$.
Since $\aut^\varphi_\beta$ and $\autB^\varphi_\beta$ share the same transition structures, the run of $\autB^\varphi_\beta$ on $u$ is also a run of $\aut^\varphi_\beta$ on $u$, which both lead to state $q$.
Therefore, $u\sigma \in \lang(\aut^\varphi_\beta)$.
By Theorem~\ref{thm:büchi}, this is equivalent to $\rltleval(u\sigma, \varphi) = \beta$.

\paragraph{From right to left:}
Let $u \in \Sigma^*$ and $\sigma \in \Sigma^\omega$ such that $\ltleval(u\sigma, \varphi) = \beta$.
By Theorem~\ref{thm:büchi}, we have $u\sigma \in \lang(\aut^\varphi_\beta)$.
Consider an accepting run of $\aut^\varphi_\beta$ on $u\sigma$, and let $q$ be the state that $\aut^\varphi_\beta$ reaches after reading the finite prefix $u$.
Since $u\sigma \in \lang(\aut^\varphi_\beta)$, this means that $\sigma \in \lang(\aut^\varphi_\beta(q))$.
Thus, $q \in F^\star_\beta$ because $\lang(\aut_\beta^\varphi(q)) \neq \emptyset$.
Moreover, since the run of $\aut^\varphi_\beta$ on $u$ is also a run of $\autB^\varphi_\beta$ on $u$, the NFA $\autB^\varphi_\beta$ can also reach state $q$ after reading $u$.
Therefore, $u \in \lang(\autB^\varphi_\beta)$ since $q \in F^\star_\beta$. 
\qed\end{proof}

Before we continue to the last step in our construction, let us briefly comment on the complexity of computing the NFAs $\autB_\beta^\varphi$.
Since $\autB_\beta^\varphi$ and $\aut_\beta^\varphi$ share the same underlying transition structure, we immediately obtain $|\autB_\beta^\varphi| \in 2^{\bigo(\card{\varphi})}$.
Moreover, the construction of $\autB_\beta^\varphi$ is dominated by the computation of the set $F_\beta^\star$ and, hence, can be done in time $2^{\bigo(\card{\varphi})}$.

\paragraph{\bfseries Step 3:}
In the final step, we construct a Moore machine implementing an \rltl monitor for $\varphi$.
Formally, a \emph{Moore machine} is a five-tuple $\mach = (Q, q_\init, \delta, \Gamma, \lambda)$ consisting of a deterministic transition structure $(Q, q_\init, \delta)$, an output alphabet~$\Gamma$, and an output function $\lambda \colon Q \to \Gamma$.
The size of~$\mach$ as well of runs of~$\mach$ are defined as for DFAs.
In contrast to a DFA, however, a Moore machine $\mach$ computes a function $\lambda_\mach \colon \Sigma^* \to \Gamma$ that is defined by $\lambda_\mach(u) = \lambda(q_n)$ where $q_n$ is the last state reached on the unique finite run~$q_0 \cdots q_n$ of $\mach$ on its input $u \in \Sigma^*$. 

The first step in the construction of the Moore machine is to determinize the NFAs~$\autB^\varphi_\beta$, obtaining equivalent DFAs~$\autC^\varphi_\beta = (Q_\beta',\allowbreak q_{\init,\beta}', \delta_\beta',\allowbreak F_\beta')$ of at most exponential size in~$\card{\autB^\varphi_\beta}$.
Subsequently, we combine these DFAs into a single Moore machine $\mach_\varphi$ implementing the desired \rltl monitor.
Intuitively, this Moore machine is the product of the DFAs $\autC^\varphi_\beta$ for each $\beta \in \bools_4$ and tracks the run of each individual DFA on the given input.
Formally, $\mach_\varphi$ is defined as follows.

\begin{definition} \label{def:moore-machine}
Let $\bools_4 = \{ \beta_1, \beta_2, \beta_3, \beta_4, \beta_5 \}$.
We define $\mach_\varphi = (Q, q_\init, \Gamma, \delta, \lambda)$ by
\begin{itemize}
	\item $Q = Q_{\beta_1}' \times Q_{\beta_2}' \times Q_{\beta_3}' \times Q_{\beta_4}' \times Q_{\beta_5}'$;
	\item $q_\init = (q_{\init,{\beta_1}}', q_{\init,{\beta_2}}', q_{\init,{\beta_3}}', q_{\init,{\beta_4}}', q_{\init,{\beta_5}}')$;
	\item $\delta \bigl( (q_1, q_2, q_3, q_4, q_5), a \bigr) = (q_1', q_2', q_3', q_4', q_5')$ where $q_j' = \delta_{\beta_j}'(q_j, a)$ for each $j \in \{1, \ldots, 5\}$;
	\item $\Gamma = \bools_4^?$; and
	\item $\lambda \bigl( (q_1, q_2, q_3, q_4, q_5) \bigr) = \xi \bigl( \bigl \{ \beta_j \in \bools_4 \mid q_j \in F_{\beta_j}', j \in \{1, \ldots, 5\} \bigr\} \bigr)$,
\end{itemize}
where the surjective function $\xi \colon 2^{\bools_4} \to \bools_4^{\ukno}$ translates sets $B \subseteq \bools_4$ of truth values to the robust monitoring semantics as follows: $\xi(B) = \beta^{\ukno} \in \bools_4^?$ with 
\[
	 \beta^{\ukno}[j] = \begin{cases} 0 & \text{if $\beta[j] = 0$ for each $\beta \in B$;} \\ 1 & \text{if $\beta[j] = 1$ for each $\beta \in B$; and} \\ \ukno & \text{otherwise.} \end{cases}
\]
\end{definition}

Let us illustrate this last step of our construction by means of our running example.

\begin{example} \label{ex:construction-step-3}
Given the NFAs $\autB^{\Boxdot s}_\beta$ from Step~2 of our construction, we first apply a standard determinization step.
This process results in equivalent DFAs $\autC^{\Boxdot s}_\beta$, which are shown in Figure~\ref{fig:ex-step-3a}.

The final, minimized monitor $\mach_{\Boxdot s}$, which results from the Cartesian product of all DFAs, is shown in Figure~\ref{fig:ex-step-3b}.
Note that this monitor has four different verdicts, shown as labels next to each state.
These are four of the verdicts used to prove results in Table~\ref{tag:truthvalues} (on Page~\pageref{tag:truthvalues}). \hfill\exampleend
\end{example}

\begin{figure}
	\centering
	\begin{subfigure}[b]{.3\textwidth}
		\centering
		\begin{tikzpicture}[auto]
			\node[state, accepting] (0) at (0, 0) {};
			\node[state] (1) at (2.25, 0) {};
			\draw[<-, shorten <=1pt] (0.west) -- +(-.3,0);
			\path[->] (0) edge[loop above] node {$s$} () edge node {$\lnot s$} (1);
			\path[->] (1) edge[loop above] node {$\mathit{true}$} ();
		\end{tikzpicture}
		\caption{The DFA $\autC^{\Boxdot s}_{1111}$}
		\label{fig:ex-step-3:1111}
	\end{subfigure}
	\hskip 2em
	\begin{subfigure}[b]{.25\textwidth}
		\centering
		\begin{tikzpicture}[auto]
			\node[state, accepting] (0) at (0, 0) {};
			\draw[<-, shorten <=1pt] (0.west) -- +(-.3,0);
			\path[->] (0) edge[loop above] node {$\mathit{true}$} ();
		\end{tikzpicture}
		\caption{The DFA $\autC^{\Boxdot s}_{0111}$}
		\label{fig:ex-step-3:0111}
	\end{subfigure}
	\hskip 2em
	\begin{subfigure}[b]{.25\textwidth}
		\centering
		\begin{tikzpicture}[auto]
			\node[state, accepting] (0) at (0, 0) {};
			\draw[<-, shorten <=1pt] (0.west) -- +(-.3,0);
			\path[->] (0) edge[loop above] node {$\mathit{true}$} ();
		\end{tikzpicture}
		\caption{The DFA $\autC^{\Boxdot s}_{0011}$}
		\label{fig:ex-step-3:0011}
	\end{subfigure}

	\bigskip
	\begin{subfigure}[b]{.25\textwidth}
		\centering
		\begin{tikzpicture}[auto]
			\node[state, accepting] (0) at (0, 0) {};
			\draw[<-, shorten <=1pt] (0.west) -- +(-.3,0);
			\path[->] (0) edge[loop above] node {$\mathit{true}$} ();
		\end{tikzpicture}
		\caption{The DFA $\autC^{\Boxdot s}_{0001}$}
		\label{fig:ex-step-3:0001}
	\end{subfigure}
	\hskip 2em
	\begin{subfigure}[b]{.3\textwidth}
		\centering
		\begin{tikzpicture}[auto]
			\node[state, accepting] (0) at (0, 0) {};
			\node[state] (1) at (2.25, 0) {};
			\draw[<-, shorten <=1pt] (0.west) -- +(-.3,0);
			\path[->] (0) edge[loop above] node {$\lnot s$} () edge node {$s$} (1);
			\path[->] (1) edge[loop above] node {$\mathit{true}$} ();
		\end{tikzpicture}
		\caption{The DFA $\autC^{\Boxdot s}_{0000}$}
		\label{fig:ex-step-3:0000}
	\end{subfigure}
	
	\caption{The DFAs $\autC^{\Boxdot s}_\beta$ constructed in Step~3 of our monitor construction}
	\label{fig:ex-step-3a}
\end{figure}

\begin{figure}
	\centering
	\begin{tikzpicture}[auto]
		\node[state] (0) at (0, 0) {};
		\node[state] (1) at (2, 1.25) {};
		\node[state] (2) at (2, -1.25) {};
		\node[state] (3) at (4, 0) {};
		\node[anchor=north east, text=black!60] at (0.south west) {$????$};
		\node[anchor=south east, text=black!60] at (1.north west) {$???1$};
		\node[anchor=north east, text=black!60] at (2.south west) {$0???$};
		\node[anchor=north west, text=black!60] at (3.south east) {$0??1$};
		\draw[<-, shorten <=1pt] (0.west) -- +(-.3,0);
		\path[->] (0) edge node {$s$} (1) edge node[swap] {$\lnot s$} (2);
		\path[->] (1) edge[loop right] node {$s$} () edge node {$\lnot s$} (3);
		\path[->] (2) edge[loop right] node {$\lnot s$} () edge node[swap] {$s$} (3);
		\path[->] (3) edge[loop above] node {$\mathit{true}$} ();
	\end{tikzpicture}

	\caption{The final monitor $\mach_{\Boxdot s}$}
	\label{fig:ex-step-3b}
\end{figure}

The main result of this paper now shows that the Moore machine $\mach_\varphi$ implements~$\rltlmoneval$, \ie, we have $\lambda_{\mach_\varphi}(u) = \rltlmoneval (u, \varphi)$ for every prefix~$u$.

\begin{theorem}\label{thm:moore}
For every \rltl formula~$\varphi$, one can construct an \rltl monitor of size~$2^{2^{\bigo(\size{\varphi})}}$.
\end{theorem}

\begin{proof}
First observe that $\xi$ indeed produces a valid value of $\bools_4^{\ukno}$ (\ie, a truth value of the form~$0^*?^*1^*$).
This follows immediately from the definition of $\xi$ and the fact that the truth values of  \rltl are sequences in $0^*1^*$.

Next, we observe that $\mach_\varphi$ reaches state $(q_1, q_2, q_3, q_4, q_5)$ after reading a word~$u \in \Sigma^*$ if and only if for each $\beta_j \in \bools_4$ the DFA $\autC^\varphi_{\beta_j}$ reaches state $q_j$ after reading $u$.
A simple induction over the length of inputs fed to $\mach_\varphi$ proves this.

Now, let us fix a word~$u \in \Sigma^*$ and assume that $(q_1, q_2, q_3, q_4, q_5)$ is the state reached by $\mach_\varphi$ after reading $u$.
This means that each individual DFA $\autC^\varphi_{\beta_j} = (Q'_{\beta_j}, q_{\init, \beta_j}', \delta'_{\beta_j}, F'_{\beta_j})$ reaches state $q_j$ after reading $u$.
Let now
\[ B = \bigl \{ \beta_j \in \bools_4 \mid q_j \in F_{\beta_j}', j \in \{1, \ldots, 5\} \bigr\} \]
as in the definition of the output function $\lambda$ of $\mach_\varphi$.
By applying Lemma~\ref{lem:B:correct}, we then obtain
\begin{align*}
	\beta_j \in B \Leftrightarrow q_j \in F'_{\beta_j} \Leftrightarrow u \in L(\mathcal C^\varphi_{\beta_j})  \Leftrightarrow u \in L(\mathcal B^\varphi_{\beta_j}) 
	\Leftrightarrow \exists \sigma \in \Sigma^\omega \colon \rltleval(u\sigma, \varphi) = \beta_j.
\end{align*}

To conclude the proof, it is left to show that $\xi(B) = \rltlmoneval (u, \varphi)$.
We show this for each bit individually using a case distinction over the elements of $\ternaries = \{ 0, ?, 1 \}$.
So as to clutter this proof not too much, however, we only discuss the case of $\ukno$ here, while noting that the remaining two cases can be proven analogously.
Thus, let $i \in \{ 1, \ldots, 4 \}$.
Then,
\begin{align*}
	\xi(B)[i] = \ukno & \Leftrightarrow \exists\beta, \beta' \in B \colon \beta[i] = 0 \text{ and } \beta'[i] = 1 \\
	& \Leftrightarrow \exists \sigma_0, \sigma_1 \in \Sigma^\omega \colon
	 \rltleval(u\sigma_0, \varphi)[i] = 0 \text{ and } \rltleval(u\sigma_1, \varphi)[i] = 1 \\
	& \Leftrightarrow \rltlmoneval (u, \varphi)[i] = \ukno.
\end{align*}

Since $\lambda\bigl( (q_1, q_2, q_3, q_4, q_5) \bigr) = \xi(B)$, the Moore machine $\mach_\varphi$ indeed outputs $\rltlmoneval (u, \varphi)$ for every word $u \in \Sigma^*$.
Moreover, $\mach_\varphi$ has $2^{2^{\bigo(\size{\varphi})}}$ states because the DFAs~$\autC^\varphi_\beta = (Q_\beta', q_{\init,\beta}',\allowbreak \delta_\beta',\allowbreak F_\beta')$ are of at most exponential size in~$\card{\autB^\varphi_\beta}$, which in turn is at most exponential in $\size{\varphi}$.
In total, this proves Theorem~\ref{thm:moore}.
\qed\end{proof}

In a final post-processing step, we minimize $\mach_\varphi$ (e.g., using one of the standard algorithms for deterministic automata).
As a result, we obtain the unique minimal monitor for the given \rltl formula.

It is left to determine the complexity of our \rltl monitor construction.
Since each DFA $\autC^\varphi_\beta$ is in the worst case exponential in the size of the NFA $\autB_\beta^\varphi$, we immediately obtain that $\autC^\varphi_\beta$ is at most of size~$2^{2^{\bigo(\card{\varphi})}}$.
Thus, the Moore machine $\mach_\varphi$ is at most of size $2^{2^{\bigo(\card{\varphi})}}$ as well and can be effectively computed in doubly-exponential time in $\card{\varphi}$.
Note that this matches the complexity bound of Bauer et al.'s approach for \ltl runtime monitoring~\cite{BauerLeuckerSchallhart11}.
Moreover, this bound is tight since \rltl subsumes \ltl (see  Remark~\ref{remark:rltlmonitoringextendsltlmonitoring}): Every monitor for an \rltl formula (without implications) can be turned into a monitor for the corresponding \ltl formula by projecting every output to its first bit. 
Thus, the doubly-exponential bound, which is tight for \ltl~\cite{DBLP:journals/fmsd/KupfermanV01,BauerLeuckerSchallhart11}, is also tight for \rltl.
Hence, robust runtime monitoring asymptotically incurs no extra cost compared to classical \ltl runtime monitoring.
However, it provides more useful information as we demonstrate next in our experimental evaluation.

\section{Experimental Evaluation}
\label{sec:evaluation}

Besides incorporating a notion of robustness into classical \ltl monitoring, our \rltl monitoring approach also promises to provide richer information than its \ltl counterpart.
In this section, we evaluate empirically whether this promise is actually fulfilled.
More precisely, we answer the following two questions on a comprehensive suite of benchmarks:
\begin{enumerate}
	\item \label{itm:research-question:1} How does \rltl monitoring compare to classical \ltl monitoring in terms of monitorability?
	\item \label{itm:research-question:2} For formulas that are both \ltl-monitorable and \rltl-mon\-i\-torable, how do both approaches compare in terms of the size of the resulting monitors and the time required to construct them?
\end{enumerate}

To answer these research questions, we have implemented a prototype, which we named \rltlmon.
Our prototype is written in Java and builds on top of two libraries: Owl~\cite{DBLP:conf/atva/KretinskyMS18}, a library for \ltl and automata over infinite words, as well as AutomataLib (part of LearnLib~\cite{DBLP:conf/cav/IsbernerHS15}), a library for automata over finite words and Moore machines.
For technical reasons (partly due to limitations of the Owl library and partly to simplify the implementation), \rltlmon uses a monitor construction that is slightly different from the one described in Section~\ref{sec:monitoring}: Instead of translating an \rltl formula into nondeterministic Büchi automata, \rltlmon constructs deterministic parity automata. 
These parity automata are then directly converted into DFAs, thus skirting the need for a detour over NFAs and a subsequent determinization step.
Note, however, that this alternative construction produces the same \rltl monitors than the one described in Section~\ref{sec:monitoring}.
Moreover, it has the same asymptotic complexity.
The sources of our prototype are available online under the MIT license.\kern-.06em\footnote{\url{https://github.com/logic-and-learning/rltl-monitoring}}

\subsection*{Benchmarks and Experimental Setup}

The starting point of our evaluation was the original benchmark suite of Bauer et al.~\cite{BauerLeuckerSchallhart11}, which is based on a survey by Dwyer on frequently used software specification patterns~\cite{DBLP:conf/icse/DwyerAC99}.
This benchmark suite consists of $97$ \ltl formulas and covers a wide range of patterns, including safety, scoping, precedence, and response patterns.
For our \rltl monitor construction, we interpreted each \ltl formula in the benchmark suite as an \rltl formula (by treating every operator as a robust operator).

We compared \rltlmon to Bauer et al.'s implementation of their \ltl monitoring approach, which the authors named \ltlmon.
This tool uses LTL2BA~\cite{DBLP:conf/cav/GastinO01} to translate \ltl formulas into Büchi automata and AT\&T's fsmlib as a means to manipulate finite-state machines.
Since LTL2BA's and Owl's input format for \ltl formulas do not match exactly, we have translated all benchmarks into a suitable format using a python script.

We conducted all experiments on an Intel Core i5-6600 @ $3.3$\,GHz in a virtual machine with $4$\,GB of RAM running Ubuntu 18.04~LTS.
As no monitor construction took longer than $600\,s$, we did not impose any time limit.

\subsection*{Results}

Our evaluation shows that \ltlmon and \rltlmon are both able to generate monitors for all $97$ formulas in Bauer et al.'s benchmark suite.\footnote{Note that the tools disagreed on one monitor where \ltlmon constructed a monitor with 1 state whereas \rltlmon constructed an \ltl monitor with 8 states.  The respective formula was removed from the reported results.}  %
Aggregated statistics of this evaluation are visualized in Figure~\ref{fig:results}.%
\footnote{Detailed results can be found in Tables~\ref{tab:res:first} and \ref{tab:res:second} in the appendix.}

\pgfplotsset{
	whisker plot/.style={
		typeset ticklabels with strut,
		x tick label style = {text depth=0pt},
		y tick style = {draw=none},
		tick align = outside,
		axis x line* = bottom,
	        axis y line* = left,
	        y axis line style = {draw=none},
	}
}

\begin{figure*}[!t]
	\begin{subfigure}[t]{\textwidth}
      \centering
		\begin{tikzpicture}
		\begin{axis}
			[
				width = .8\textwidth,
				height = 45mm,
				ybar,
				ymin = 0,
				ymax = 50,
				axis on top,
				ymajorgrids,
				typeset ticklabels with strut,
				x tick label style = {text depth=0pt},
			        major grid style = {draw=white},
			        tick style = {draw=none},
			        xtick distance = 1,
				axis x line* = bottom,
			        axis y line* = left,
			        y axis line style = {draw=none},
			        nodes near coords,
				xlabel = {Number of states},
				ylabel= {Number of monitors \\ ($97$ in total)},
				ylabel style={align = center, yshift=1.5ex},
				legend style = {draw=none, column sep=.75em, row sep=1ex},
				legend image post style = {draw opacity=0, scale=1.25},
				legend cell align = {left},
			]

			\addplot [draw=none, fill=gray!60] coordinates {(1, 43) (2, 20) (3, 21) (4, 11) (5, 1) (6, 1) (7, 0) (8, 0)};
			\addplot [draw=none, fill=gray!25] coordinates {(1, 0) (2, 39) (3, 10) (4, 29) (5, 9) (6, 7) (7, 1) (8, 2)};

			\legend {\ltlmon, \rltlmon};	
		\end{axis}
		\end{tikzpicture}
		\caption{Histogram of the number of monitors with respect to their size} \label{subfig:results:histogram}		
	\end{subfigure}
	\vskip 1.5\baselineskip
	\begin{subfigure}[t]{\textwidth}
		\centering
		\begin{tikzpicture}
			\begin{axis}
				[
					whisker plot,
					width = .35\textwidth,
					height = 11mm,
					scale only axis = true,
					xmin = 0,
					xmax = 9,
					ytick = {1, 2},
					yticklabels = {\rltlmon, \ltlmon},
					xlabel = {Number of states},
				]
		
				\addplot[draw=gray, fill=gray!25, boxplot prepared={lower whisker=2, lower quartile=4, median=4, upper quartile=5, upper whisker=8}] coordinates {};
				\addplot[draw=gray, fill=gray!60, boxplot prepared={lower whisker=2, lower quartile=2, median=3, upper quartile=3, upper whisker=6}] coordinates {};
			\end{axis}
		\end{tikzpicture}
		\hskip 1em
		\begin{tikzpicture}
			\begin{axis}
				[
					whisker plot,
					width = .35\textwidth,
					height = 11mm,
					scale only axis = true,
					xmode = log,
					xmin = 1e-3,
					xmax = 100,
					ymajorticks = false,
					xlabel = {Time in $s$},
				]
		
				\addplot[draw=gray, fill=gray!25, boxplot prepared={lower whisker=0.82, lower quartile=1.125, median=1.47, upper quartile=2.115, upper whisker=35.57}] coordinates {};
				\addplot[draw=gray, fill=gray!60, boxplot prepared={lower whisker=0.01, lower quartile=0.01, median=0.015, upper quartile=0.02, upper whisker=2.11}] coordinates {};
			\end{axis}
		\end{tikzpicture}
		\caption{Analysis of the monitor construction for the $54$ formulas that are both \ltl-monitorable and \rltl-monitorable}
		\label{subfig:results:whisker-plot}	
	\end{subfigure}	
  \caption{Comparison of \rltlmon and \ltlmon on Bauer et al.'s benchmarks~\cite{BauerLeuckerSchallhart11}\label{fig:results}}

\end{figure*}
The histogram in Figure~\ref{subfig:results:histogram} shows the aggregate number of \ltl and \rltl monitors with respect to their number of states.
As Bauer et al.\ already noted in their original work, the resulting \ltl monitors are quite small (none had more than six states), which they attribute to Dwyer et al.'s specific selection of formulas~\cite{DBLP:conf/icse/DwyerAC99}.
A similar observation is also true for the \rltl monitors: None had more than eight states.

To determine which formulas are monitorable and which are not, we used a different, though equivalent definition, which is easy to check on the monitor itself:
an \ltl formula (\rltl formula) is monitorable if and only if the unique minimized \ltl monitor (\rltl monitor) does not contain a sink-state with universal self-loop that outputs ``$\ukno$'' (that outputs ``$\ukno\ukno\ukno\ukno$''). In other words, even if a finite word does not allow us to infer anything about the satisfaction of the \ltl (\rltl) formula by infinite words extending it, it can always be extended into another finite word that does.
Bauer \etal report that $44.3\,\%$ of all \ltl monitors ($43$ out of $97$) have this property (in fact, exactly the $43$ \ltl monitors with a single state), which means that $44.3\,\%$ of all formulas in their benchmark suite are not \ltl-monitorable.
By contrast, all these formulas are \rltl-monitorable.
Moreover, in $78.4\,\%$ of the cases ($76$ out of $97$), the \rltl monitor has more distinct outputs than the \ltl monitor, indicating that the \rltl monitor provides more fine-grained information of the property being monitored; in the remaining $21.6\,\%$, both monitors have the same number of distinct outputs.
These results answer our first research question strongly in favor of \rltl monitoring: \emph{\rltl monitoring did in fact provide more information than its classical \ltl counterpart. In particular, only $55.7\,\%$ of the benchmarks are \ltl-monitorable, whereas $100\,\%$ are \rltl-monitorable}.

Let us now turn to our second research question and compare both monitoring approaches on the $54$ formulas that are both \ltl-monitorable and \rltl-monitorable.
For these formulas, Figure~\ref{subfig:results:whisker-plot} further provides statistical analysis of the generated monitors in terms of their size (left diagram) as well as the time required to generate them (right diagram).
Each box in the diagrams shows the lower and upper quartile (left and right border of the box, respectively), the median (line within the box), and minimum and maximum (left and right whisker, respectively).

Let us first consider the size of the monitors (left diagram of Figure~\ref{subfig:results:whisker-plot}).
The majority of \ltl monitors ($52$) has between two and four states, while the majority of \rltl monitors ($45$) has between two and five states.
For $21$ benchmarks, the \ltl and \rltl monitors are of equal size, while the \rltl monitor is larger for the remaining $33$ benchmarks.
On average, \rltl monitors are about $1.5$ times larger than the corresponding \ltl monitors.

Let us now discuss the time taken to construct the monitors.
As the diagram on the right-hand-side of Figure~\ref{subfig:results:whisker-plot} shows, \ltlmon was considerably faster than \rltlmon on a majority of benchmarks (around $0.1\,s$ and $2.6\,s$ per benchmark, respectively).
For all $54$ benchmarks, the \rltl monitor construction took longer than the construction of the corresponding \ltl monitor (although there are two non-\ltl-monitorable formulas for which the construction of the \rltl monitor was faster).
However, we attribute this large runtime gap partly to the overhead caused by repeatedly starting the Java virtual machine, which is not required in the case of \ltlmon.
Note that this is not a concern in practice as a monitor is only constructed once before it is deployed.

Finally, our analysis answers our second question: \emph{\rltl monitors are only slightly larger than the corresponding \ltl monitors and although they require considerably more time to construct, the overall construction time was negligible for almost all benchmarks}.

\section{Conclusion}

We adapted the three-valued \ltl monitoring semantics of Bauer \etal to \rltl, proved that the construction of monitors is asymptotically no more expensive than the one for \ltl, and validated our approach on the benchmark of Bauer \etal: All formulas are \rltl-monitorable and the \rltl monitor is strictly more informative than its \ltl counterpart for 77\% of their formulas.

Recall Theorem~\ref{thm:realizabletruthvalues}, which states that the truth values $0011$ and $0001$ are not realizable. 
This points to a drawback regarding the two middle bits: 
When considering the formula $\Boxdot a$, the second bit represents $\Diamond\Box a$ and the third bit $\Box\Diamond a$. 
A prefix cannot possibly provide enough information to distinguish these two formulas.
On the other hand, the truth value $??11$ is realizable, which shows that the middle bits can be relevant. In further work, we will investigate the role of the middle bits in \rltl monitoring.

Moreover, the informedness of a monitor can be increased further when attributing a special role to the last position(s) of a prefix. 
Even though a prefix of the form $\emptyset^+\set{a}^+$ does not fully satisfy $\Diamond\Box a$, neither does it fully violate it. 
If the system just now reached a state in which $\set{a}$ always holds, an infinite continuation of the execution would satisfy the specification.
So rather than reporting an undetermined result, the monitor could indicate that an infinite repetition of the last position of the prefix would satisfy the specification. 
Similarly, for a prefix $\set{a}^+\emptyset$, the specification $\Box\Diamond a$ is undetermined. 
While an infinite repetition of the last position ($\set{a}^+\emptyset^\omega$) does not satisfy the specification, an infinite repetition of the last two positions ($\set{a}^+(\emptyset\set{a})^\omega$) would. 
We plan to investigate an extension of \rltl which takes this observation into account.

Bauer~\etal~\cite{DBLP:conf/rv/0002LS07} proposed an orthogonal approach with the logic RV-LTL. 
Here, the specification can contain the strong (weak) next-operator whose operand is consider violated (satisfied) at the last position of the trace.
A formula that is undetermined under the strong semantics and satisfied (violated) under the weak semantics is considered \emph{potentially true (potentially false)}.
Incorporating one of these approaches into \rltl monitoring could refine its output and thus increase its level of informedness.

Moreover, desired properties for cyber-physical systems often include real-time components such as ``touch the ground at most 15 seconds after receiving a landing command''. Monitors for logics taking real-time into account~\cite{DBLP:conf/sosp/BernsteinH81}, such as STL~\cite{DBLP:conf/formats/MalerN04,DBLP:conf/birthday/MalerNP08}, induce high computational overhead at runtime when compared to \ltl and \rltl monitors. Thus, a real-time extension for \rltl retaining its low runtime cost would greatly increase its viability as specification language. 

\subsection*{Acknowledgements}
The authors would like to thank Li Bingchen for discovering the formula mentioned in Footnote \ref{bingchenfootnote}.
The work of Daniel Neider was supported by the {Deutsche Forschungsgemeinschaft (DFG, German Research Foundation)} -- grant no.\ 434592664.
The work of Maximilian Schwen\-ger was supported by the {European Research Council (ERC) Grant OSARES} (No.~{683300}) and the {German Research Foundation} (DFG) as part of the Collaborative Research Center ``Center for Perspicuous Computing'' (TRR 248,~{389792660}).
The work of Paulo Tabuada was partially supported by the {NSF project} {1645824}. 
The work of Alexander Weinert was supported by the {Saarbrücken Graduate School of Computer Science}.
The work of Martin Zimmermann was supported by the {Engineering and Physical Sciences Research Council (EPSRC)} project {EP/S032207/1}.

\bibliographystyle{plainurl}
\bibliography{bib}

\clearpage
\appendix
\input{appendix}
\end{document}

%% file: preamble.tex
\usepackage[utf8]{inputenc}
\usepackage[T1]{fontenc}
\usepackage{microtype}
\usepackage{amssymb}

\usepackage{mathtools}
\usepackage[inline]{enumitem}
\usepackage{nicefrac}
\usepackage{booktabs, multirow}
\usepackage{tikz}
\usepackage{pgfplots}
\usepackage{underscore}
\usepackage{multirow}
\usepackage{graphicx}

\usepackage{csvsimple} 

\usepackage{csquotes}
\usepackage{hyperref}
\usepackage{xspace}
\usepackage{subcaption}
\usepackage{centernot}
\usepackage{color}

\usepackage{xspace}
\usepackage{array}
\usepackage{xparse}
\usepackage{pifont}

\newcommand{\ltl}{\text{LTL}\xspace}
\newcommand{\rltl}{\text{rLTL}\xspace}

\newcommand{\ltlformulas}{\ensuremath{\Phi_{\kern-1pt\mathit{LTL}}}}
\newcommand{\rltlformulas}{\ensuremath{\Phi_{\kern-1pt\mathit{rLTL}}}}

\newcommand{\ltleval}{V}
\newcommand{\rltleval}{V_r}
\newcommand{\rltlalteval}{\rltleval}
\newcommand{\ltlmoneval}{V^m }
\newcommand{\rltlmoneval}{V_r^m }

\newcommand{\rltlsem}{\mathrm{ltl}}

\newcommand{\suff}[2]{#1[#2,\infty)}
\newcommand{\rot}[2]{\rho(#1,#2)}

\newcommand{\set}[1]{\{ #1 \}}
\newcommand{\card}[1]{| #1 |}
\newcommand{\size}[1]{| #1 |}
\newcommand{\lang}{{\mathcal L}}
\newcommand{\bigo}{{\mathcal O}}

\newcommand\ukno{\mathord{?}}

\newcommand{\bools}{\mathbb{B}}
\newcommand{\ternaries}{\mathbb{B}^{\ukno}}

\newcommand{\aut}{\ensuremath{\mathcal A}}
\newcommand{\autB}{\ensuremath{\mathcal B}}
\newcommand{\autC}{\ensuremath{\mathcal C}}
\newcommand{\mach}{\ensuremath{\mathcal M}}
\newcommand{\tsys}{\ensuremath{\mathcal T}}
\newcommand{\init}{I}

\newcommand{\true}{\ensuremath{\top}}
\newcommand{\false}{\ensuremath{\bot}}

\renewcommand{\implies}{\rightarrow}

\newcommand{\ltlmon}{\texttt{LTL\textsubscript{3}~tools}\xspace}
\newcommand{\rltlmon}{\texttt{rLTL-mon}\xspace}

\newcolumntype{L}{>{$}l<{$}}
\newcolumntype{R}{>{$}r<{$}}
\newcolumntype{C}{>{$}c<{$}}

\tikzset{robust/.style={line width=.16ex,line join=round}}

\let\Box\relax
\DeclareMathOperator{\Box}{%
	\text{%
		\tikz[baseline]{%
    			\draw[robust] (0ex,-.1ex) -- (0ex, 1.4ex) -- (1.5ex, 1.4ex) -- (1.5ex, -.1ex) -- cycle;%
		}%
	}%
}

\DeclareMathOperator{\Boxdot}{%
	\text{%
		\tikz[baseline]{%
    			\draw[robust] (0ex, -.1ex) -- (0ex, 1.4ex) -- (1.5ex, 1.4ex) -- (1.5ex, -.1ex) -- cycle;%
	    		\fill (.75ex, .65ex) circle (.15ex);%
    		}%
	}%
}

\let\Diamond\relax
\DeclareMathOperator{\Diamond}{%
	\text{%
		\tikz[baseline]{%
			\draw[robust] (0ex,.6ex) -- (.95ex, 1.55ex) -- (1.9ex, .6ex) -- (.95ex, -.35ex) -- cycle;%
		}%
	}%
}

\let\Diamonddot\relax
\DeclareMathOperator{\Diamonddot}{%
	\text{%
		\tikz[baseline]{%
			\draw[robust] (0ex,.6ex) -- (.95ex, 1.55ex) -- (1.9ex, .6ex) -- (.95ex, -.35ex) -- cycle;%
			\fill (.95ex, .6ex) circle (.15ex);%
		}%
	}%
}

\DeclareMathOperator{\R}{%
	\text{%
		\tikz[baseline]{%
			\node[inner sep=0pt, anchor=base, font=\normalfont\bfseries] {R};%
		}%
	}%
}

\DeclareMathOperator{\Rdot}{%
	\text{%
		\tikz[baseline]{%
			\node[inner sep=0pt, anchor=base, font=\normalfont\bfseries] {R};%
			\fill (-.025ex, 1.175ex) circle (.15ex);%
		}%
	}%
}

\DeclareMathOperator{\Rimplies}{%
  \implies
}

\let\U\relax
\DeclareMathOperator{\U}{%
	\text{%
		\tikz[baseline]{%
			\node[inner sep=0pt, anchor=base, font=\normalfont\bfseries] {U};%
		}%
	}%
}

\DeclareMathOperator{\Udot}{%
	\text{%
		\tikz[baseline]{%
			\node[inner sep=0pt, anchor=base, font=\normalfont\bfseries] {U};%
			\fill (.1ex, .9ex) circle (.15ex);%
		}%
	}%
}

\DeclareMathOperator{\X}{%
	\text{%
		\tikz[baseline]{%
   			\draw[robust] (.75ex, .65ex) circle (.75ex);%
    		}%
	}%
}

\DeclareMathOperator{\Xdot}{%
	\text{%
		\tikz[baseline]{%
    			\draw[robust] (.75ex, .65ex) circle (.75ex);%
	    		\fill (.75ex, .65ex) circle (.15ex);%
    		}%
	}%
}

\NewDocumentCommand{\threepartdef}{m m m m m O{\text{otherwise}}}{
  \left\{
    \begin{array}{lll}
      #1 & \quad\mbox{if } #2 \\
      #3 & \quad\mbox{if } #4 \\
      #5 & \quad\mbox{if } #6
    \end{array}
  \right.
}

\makeatletter
\newcommand*{\etal}{%
    \@ifnextchar{.}%
        {et~al}%
        {et~al.\@\xspace}%
}
\newcommand*{\ie}{i.e.\@\xspace}
\newcommand*{\eg}{e.g.\@\xspace}
\makeatother

\pgfplotsset{compat=1.14}
\usepgfplotslibrary{statistics}

\newcommand{\exampleend}{%
	\begin{tikzpicture}%
		\draw (0, 0) -| ++(1.25ex, 1.25ex); %
	\end{tikzpicture}%
}

\usetikzlibrary{automata, arrows}
\tikzset{every edge/.append style={shorten >=1pt}, >=stealth}

%% file: appendix.tex



\section{Experimental Results}

The following two tables provide detailed results of our experimental evaluation.

\begin{table*}
\begin{tabular}{l*{8}{r}}
	\toprule
	& \multicolumn{2}{c}{\# States} & \multicolumn{2}{c}{\# Outputs} & \multicolumn{2}{c}{Monitorable} & \multicolumn{2}{c}{Time in s} \\ \cmidrule(lr){2-3} \cmidrule(lr){4-5} \cmidrule(lr){6-7} \cmidrule(lr){8-9}
	Property & \rltl & \ltl & \rltl & \ltl & \rltl & \ltl & \rltl & \ltl \\ \midrule
	\csvreader[head to column names, after line={\\}]{resultslow-generated.csv}{}{\Property&\rltlstates&\ltlstates&\rltloutputs&\ltloutputs&\rltlmon&\ltlmon&\rltltime&\Ltltime}
	\\[-2ex]\bottomrule
\end{tabular}
\caption{Summary of the result when comparing the monitor construction of \rltl against \ltl; continued in Table~\ref{tab:res:second}}
\label{tab:res:first}
\end{table*}

\begin{table*}
\begin{tabular}{l*{8}{r}}
	\toprule
	& \multicolumn{2}{c}{\# States} & \multicolumn{2}{c}{\# Outputs} & \multicolumn{2}{c}{Monitorable} & \multicolumn{2}{c}{Time in s} \\ \cmidrule(lr){2-3} \cmidrule(lr){4-5} \cmidrule(lr){6-7} \cmidrule(lr){8-9}
	Property & \rltl & \ltl & \rltl & \ltl & \rltl & \ltl & \rltl & \ltl \\ \midrule
	\csvreader[head to column names, after line={\\}]{resultshigh-generated.csv}{}{\Property&\rltlstates&\ltlstates&\rltloutputs&\ltloutputs&\rltlmon&\ltlmon&\rltltime&\Ltltime}
	\\[-2ex]\bottomrule
\end{tabular}
\caption{Summary of the result when comparing the monitor construction of \rltl against \ltl; continuation of Table~\ref{tab:res:first}}
\label{tab:res:second}
\end{table*}